\documentclass[preprint,12pt]{elsarticle1}

\usepackage{subfig}
\usepackage{graphicx}
\usepackage{caption,color}   
\usepackage{amssymb}
\usepackage{amsthm}
\usepackage{amsmath}
\usepackage{epic}
\usepackage{setspace}
\usepackage{float}
\usepackage{multirow}

\usepackage{hyperref}
\usepackage{xcolor}
\usepackage{marginnote}
\usepackage{booktabs}
\usepackage{algorithm, algorithmicx, algpseudocode}

\newtheorem{thm}{Theorem}[section]

\newtheorem{lem}[thm]{Lemma}

\newtheorem{defi}[thm]{Definition}
\newtheorem{remark}[thm]{Remark}
\newtheorem{example}[thm]{Example}

\numberwithin{equation}{section}
\journal{}

\begin{document}
\begin{spacing}{1.15}
\begin{frontmatter}
\title{\textbf{The eigenvector centrality of hypergraphs}}

\author{Changjiang Bu}\ead{buchangjiang@hrbeu.edu.cn}
\author{Haotian Zeng}
\author{Qingying Zhang}

\address{School of Mathematical Sciences, Harbin Engineering University, Harbin, PR China}

\begin{abstract}
A hypergraph is called uniform when every hyperedge contains the same number of vertices, otherwise, it is called non-uniform.
In the real world, many systems give rise to non-uniform hypergraphs, such as email networks and co-authorship networks.
A uniform hypergraph has a natural one-to-one correspondence with its adjacency tensor.
In 2019, Benson proposed the eigenvector centrality of uniform hypergraphs via its adjacency tensor \cite{benson2019three}.
In this paper, we define an adjacency tensor for hypergraphs and propose the eigenvector centrality for hypergraphs.
When the hypergraph is uniform, our proposed eigenvector centrality reduces to Benson's.
When each edge of the uniform hypergraph contains exactly two vertices, our proposed centrality reduces to the eigenvector centrality of graphs.
We conducted experiments on several real-world hypergraph datasets.
The results show that, compared to traditional centrality measures, the proposed centrality measure provides a unique perspective for identifying important vertices and can also effectively identify them.
\end{abstract}


\end{frontmatter}

\section{Introduction}
With the development of complex network analysis, traditional graph models have gradually shown limitations in capturing group interactions, collaborative relationships, and higher-order dependencies.
Hypergraph is a natural extension of graph, allowing an edge to contain any number of vertices, thereby enabling a more direct and flexible representation of the higher-order relationships commonly found in the real world.
A hypergraph where every hyperedge contains exactly $k$ vertices is called a $k$-uniform hypergraph.
When $k=2$, such a hypergraph reduces to an ordinary graph. If hyperedges in a hypergraph have varying vertex counts, it is termed a non-uniform hypergraph.
Compared to graphs, hypergraphs can model multi-way relationships, and they naturally find applications in email networks \cite{comrie2021hypergraph}, co-authorship networks \cite{lung2018hypergraph}, transportation networks \cite{prakash2017finding}, drug networks \cite{saifuddin2023hygnn} and so on.

In recent years, research on hypergraph centrality has attracted considerable attention.
Some researchers have extended the study of vertex centrality in graphs to hypergraphs.
This includes classic centrality measures such as degree centrality \cite{kapoor2013weighted}, betweenness centrality \cite{lee2021betweenness}, closeness centrality \cite{aksoy2020hypernetwork,amato2018centrality}.
Additionally, some studies identify significant vertices in hypergraphs from alternative perspectives \cite{amato2018centrality,kovalenko2022vector,vasilyeva2024matrix,tejaswi2024identifying,hu2023identifying}.

Many approaches assess vertex importance through metrics such as degree and shortest paths, yet they often fail to adequately account for the influence of neighboring vertices.
In 1972, Bonacich proposed the eigenvector centrality of graphs \cite{bonacich1972factoring}, and defined the principal eigenvector of the graph's adjacency matrix as the centrality scores of its vertices.
A vertex's importance is proportional to the sum of the importance of its neighbors.
It is widely applied in Google's PageRank search engine \cite{langville2005survey},
social network analysis \cite{howlader2016degree},
brain science networks \cite{lohmann2010eigenvector}, and so on.
To more comprehensively reflect a vertex's global influence within the network, Benson proposed a definition of eigenvector centrality for uniform hypergraphs \cite{benson2019three}, deriving the eigenvector from the adjacency tensor of the hypergraph.
However, most real-world systems correspond to non-uniform hypergraphs.
For example, in co-authorship networks, the number of authors per paper varies, while in email exchange networks, the number of recipients per email also differs.
Although some centrality measures for non-uniform hypergraphs have been proposed, these methods are not grounded in the mathematical framework of eigenvectors. Thus existing methods cannot capture the mutual influence among vertices or the global iterative nature inherent in eigenvector centrality.

In this paper, we propose an eigenvector centrality  for hypergraphs.
The vertex's centrality in a hypergraph is governed by the sum of geometric mean of the scores of all vertices sharing the same hyperedge.
When the hypergraph is uniform, our proposed centrality reduces to the centrality introduced by Benson \cite{benson2019three}.
When each edge of the uniform hypergraph contains exactly two vertices, our proposed centrality reduces to the eigenvector centrality of graphs.
%

This paper is organized as follows.
In Section 2, we introduce the necessary preliminaries for tensors and hypergraphs.
In Section 3, we introduce the least common multiple associated with the hypergraph's cardinality set to construct a unified tensor representation.
This tensor naturally captures structural information from hyperedges of varying sizes and transforms the hypergraph eigenvector centrality problem into an eigenvalue problem of this tensor.
We prove that when the hypergraph is connected, the tensor is weakly irreducible, thereby ensuring the existence and uniqueness (up to scaling) of a positive eigenvector.
An algorithm for computing the centrality proposed in this paper is provided in Section 4, and experiments were performed on real-world hypergraphs.
In Section 5, we summarize the conclusions of this paper and discusses potential directions for future work.

\section{Preliminaries }
For a positive integer $n$, let $[n]=\{1,2,\cdots,n\}$. A $k$-th order $n$-dimensional tensor $\mathcal{A}=(a_{{i_1} {i_2} {\cdots}  {i_k}}) $ is a multi-dimensional array containing $n^k$ elements, where each index $i_j \in [n]$.
 Let  $\mathbb{C}^{[k,n]}$ ($\mathbb{R}^{[k,n]}$) denote the set of $k$-th order $n$-dimensional complex (real) tensors.
For $\mathcal{A}=(a_{{i_1} {i_2} {\cdots}  {i_k}}) \in \mathbb{C}^{[k,n]}$ and $\boldsymbol {x}=(x_1,x_2,{\cdots},x_n)^{\mathsf{T}} \in \mathbb{C}^n$,
$\mathcal{A}{\boldsymbol {x}}^{k-1}$ is an $n$-dimensional vector ,
with the $i$-th component
\begin{align*}
(\mathcal{A}{\boldsymbol {x}}^{k-1})_i=\sum_{{i_2}, {\cdots},  {i_k}=1}^{n}  a_{{i} {i_2} {\cdots} {i_k}} {x_{i_2}}{x_{i_3}}{\cdots}{x_{i_k}}\quad (i \in [n]).
\end{align*}
If there exists  $\lambda \in \mathbb{C}$ and a nonzero vector $\boldsymbol {x}$ such that
\begin{align*}
\mathcal{A}{\boldsymbol {x}}^{k-1}=\lambda \boldsymbol {x}^{[k-1]},
\end{align*}
where $\boldsymbol {x}^{[k-1]}=(x_1^{k-1}, {\cdots}, x_n^{k-1})^{\mathsf{T}}$.
Then $\lambda$  is called an eigenvalue of $\mathcal{A}$,
and $\boldsymbol {x}$ is an eigenvector of  $\mathcal{A}$ corresponding to $\lambda$
\cite{lim2005singular,qi2005eigenvalues}.
The spectral radius $\rho (\mathcal{A})=\max\{|\lambda|:\lambda \in \sigma(\mathcal{A})\}$, where $\sigma(\mathcal{A})$ is the set of all eigenvalues of $\mathcal{A}$.
The eigenvalues of tensors have attracted extensive attention since they have been proposed \cite{chen2024spectra,cooper2012spectra,qi2017tensor,sun2016moore}.
And they have  been applied to hypergraph clustering \cite{chang2020hypergraph}, network traffic anomaly detection \cite{fan2024novel},
crystallography \cite{chen2023c}, mechanics \cite{nikabadze2016eigenvalue}, and so on.

For a  tensor $\mathcal{A}=(a_{{i_1} {i_2} {\cdots}  {i_k}})
\in \mathbb{R}^{[k,n]}$,
let $D_{\mathcal{A}}=(V{(D_{\mathcal{A}})},E{(D_{\mathcal{A}})})$ be the associated directed graph of $\mathcal{A}$,
with the vertex set $V{(D_{\mathcal{A}})}=[n]$, and the arc set
$E{(D_{\mathcal{A}})}=\left\{ {(i,j)| a_{i {i_2}{\cdots}{i_k} } \neq 0 ,  j \in {\left\{{{i_2},{\cdots},{i_k}}\right\}}  }\right\}$ \cite{qi2017tensor}.
$D_{\mathcal{A}}$  is said to be strongly connected if there exists a directed path from  $i$ to $j$ for any distinct $i,j \in {V(D_{\mathcal A})}$.

A tensor $\mathcal{A}=(a_{{i_1} {i_2} {\cdots}  {i_k}}) \in \mathbb{C}^{[k,n]} $ is called weakly reducible if there
exists a nonempty proper index subset $I \subseteq [n]$ such that  $a_{{i_1} {i_2} {\cdots}  {i_k}} = 0$ for $ \forall i_1 \in I
$ and $\{i_2,\cdots,i_k\} \not\subset I $.
If $\mathcal{A}$ is not weakly reducible, then $\mathcal{A}$ is called as weakly irreducible \cite{friedland2013perron}.
The following two lemmas on weakly irreducible tensors are needed for this paper.

\begin{lem}\cite{friedland2013perron}\label{connect}
	The  tensor $\mathcal{A}\in \mathbb{R}^{[k,n]}$ is weakly irreducible if and only if $D_{\mathcal{A}}$ is strongly connected.
\end{lem}
%

\begin{lem} \cite{friedland2013perron} \label{perron}
	Let $\mathcal{A}\in \mathbb{R}^{[k,n]}$ be a nonnegative weakly irreducible tensor, then the spectral radius $\rho (\mathcal{A})$ is an eigenvalue of $\mathcal{A}$, and there exists a unique positive eigenvector corresponding to $\rho (\mathcal{A})$ (up to its scalar multiples).
\end{lem}

For a hypergraph $H=(V,E)$, the rank of $H$ is the maximum cardinality of all hyperedges in $H$, i.e. $rank(H)=max\{|e|:e \in E\}$. The set of all hyperedges of $k$-cardinality is denoted by $E_k$, and the set of all hyperedges contained the vertex  $i$ of $k$-cardinality  is denoted by $E_k(i)$.



\section{Eigenvector centrality of hypergraphs}
Recall the eigenvector centrality of graphs firstly. For a  connected graph $G=(V,E)$ with $n$ vertices,  let $A_G=(a_{ij})$ be its adjaceny matrix and $x_i$ denote the eigenvector centrality score of vertex $i$ $(i\in[n])$. We konw the eigenvector centrality of each vertex is proportional to the sum of the scores of its neighbors. Let the positive scalar $\frac{1}{\lambda} >0$ denote the proportionality constant.\\
i.e.
\[
x_{i}=\frac{1}{\lambda}\sum_{\{i, j\}\in E}x_{j}=\frac{1}{\lambda}\sum_{j=1}^{n} a_{ij}x_{j}=\frac{1}{\lambda}(A_{G}x)_{i}\quad(i\in[n]).
\]

Next, we generalize the eigenvector centrality from graphs to hypergraphs.
The centrality of a vertex is proportional to the sum of the geometric means of the scores of the other vertices in every hyperedge containing the vertex.

Let $ H$ be a connected hypergraph with $n$ vertices and $ rank(H)=m$.
Let the positive scalar $\frac{1}{\lambda} >0$ denote the proportionality constant and $\boldsymbol {x}=(x_1,\ldots,x_n)^{\mathsf{T}} $ be the centrality score vector of  all vertices of $H$, where $ x_{i}>0$. We have $\lambda>0$
and $\boldsymbol {x}$ satisfy the following equation :

\begin{align}\label{1}
x_i &= \frac{1}{\lambda}\Bigl( \sum_{\{i,i_{1}\}\in E_2}x_{i_{1}} + \sum_{\{i,i_{1},i_{2}\}\in E_3}(x_{i_{1}}x_{i_{2}})^{\frac{1}{2}} +\cdots+ \sum_{\{i,i_{1},\cdots,i_{m-1}\}
	\in E_m}(x_{i_{1}}x_{i_{2}}\cdots x_{i_{m-1}})^{\frac{1}{m-1}}  \Bigr)\\
&= \frac{1}{\lambda}\sum_{k=2}^{m}\;\sum_{e\in E_k(i)} \bigl( \boldsymbol{x}^{e\setminus\{i\}} \bigr)^{\frac{1}{k-1}},
\end{align}
where $i\in[n]$, $\boldsymbol {x}^{e\setminus \{i\}}:=\prod_{j \in {e\setminus \{i\}}}x_j$.

The cardinality set of $H$ is the set of different cardinality number of all hyperedges in $H$.
Before presenting the definitions of the adjacency tensor and eigenvector centrality, we first introduce the following definition.

\begin{defi}
	For a hypergraph $H$ with  cardinality set $\{l_1,l_2,\cdots,l_{r}\}$,
	let $s(H)= \mathrm{lcm}(l_1-1, l_2-1, \cdots,l_r-1)$.
\end{defi}

To make the degree of each variable in Equation \eqref{1} an integer, when $s(H)=s$, let $\boldsymbol {y}=(y_1,\ldots,y_n)^{\mathsf{T}} $, where $ y_{i}=\sqrt[s]{x_{i}}$ for $i\in[n]$ . Then we have

\begin{equation}\label{2}
	 y_i^{s} = \frac{1}{\lambda}\sum_{k=2}^{m}\sum_{{e}\in E_k(i)} (  \boldsymbol {y}^{e\setminus \{i\}})^{\frac{s}{k-1}}\quad(i\in[n]),
\end{equation}
where $\boldsymbol {y}^{e\setminus \{i\}}:=\prod_{j \in {e\setminus \{i\}}}y_j$.

\begin{example}
	Let $ H$ be a connected hypergraph with  $n$ vertices and cardinality set $\{2,3,4\}$, we
	have
	\begin{equation*}
		 x_i = \frac{1}{\lambda}\sum_{\{i,j\}\in E_2(i)}x_j+\frac{1}{\lambda}\sum_{\{i,j,k\}\in E_3(i)}\sqrt{x_jx_k}+\frac{1}{\lambda}\sum_{\{i,j,k,l\}\in E_4(i)}\sqrt[3]{x_jx_kx_l}\quad(i\in[n]).
	\end{equation*}
	Observe $s(H) = \mathrm{lcm}(1, 2, 3)=6$. Let $ y_{i}=\sqrt[6] x_{i} $ for $i\in [n]$. Then we have
	\begin{equation*}
		 y_i^6 = \frac{1}{\lambda}\sum_{\{i,j\}\in E_2(i)}y_j^{6}+\frac{1}{\lambda}\sum_{\{i,j,k\}\in E_3(i)}y_j^{3}y_k^{3}+\frac{1}{\lambda}\sum_{\{i,j,k,l\}\in E_4(i)}y_j^{2}y_k^{2}y_l^{2}\quad(i\in[n]).
	\end{equation*}
	
\end{example}

The definition of adjacency tensor for hypergraphs is given below.

\begin{defi} \label{defi11}
	For a hypergraph $H$ with $n$ vertices, $rank(H)=m$ and $s(H)=s$, its adjacency tensor $\mathcal{A}_{H}=(a_{i_{1}i_{2}\cdots i_{s+1}})$ is an $s+1$-order $n$-dimensional tensor.
	
	For all hyperedges $ e\in E $,
	$$
	a_{i_{1}i_{2}\cdots i_{s+1}} =\frac{[(\frac{s}{k-1})!]^{k-1}}{s!},
	$$
	where $k=|e|\in\{2,3,\cdots,m\}$,  $i_{1}i_{2}\cdots i_{s+1}$ are all
	index sequences that $i_{1}\in e$ and the multiset  $\{i_{2},\cdots ,i_{s+1}\}$ contains each element of $ e\setminus \{i_1\}$ with equal  $\frac{s}{k-1}$ times. The other postions of the tensor are zeros.
\end{defi}

Later, we prove that Equation \eqref{2} has a unique positive solution $\boldsymbol {y}$ for $\lambda=\rho(\mathcal{A}_H)$ when hypergraph $H$ is connected (see Theorem \ref{dingli}).
\begin{remark}
	For an $r$-uniform hypergraph $H$, observe $s(H)=r-1$ obviously. Then
	$$
	a_{i_{1}i_{2}\cdots i_{r}} =
	\begin{cases}
		\frac{1}{(r-1)!}, & \{i_{1},i_{2},\cdots ,i_{r}\}\in E_{{r}}(H),  \\
		0, & \text{otherwise}.
	\end{cases}
	$$
\end{remark}
 And the necessary and sufficient condition for the weak irreducibility of  the adjacency tensor can be given below.

 \begin{thm} \label{liantong}
 	For a hypergraph  $H$, the adjacency tensor $\mathcal{A}_H$  is weakly irreducible if and only if $H$ is connected.
 \end{thm}
\begin{proof}
	By Lemma \ref{connect}, we only need to prove $D_{\mathcal{A}_H}$ is strongly connected if and only if $H$ is connected. Notice that $V(H)=V(D_{\mathcal{A}_H})$.
	
	When $H$ is connected. For $i, j \in V(H)$, $i \neq j$, there exists a path $P_s = i_1 e_1 i_2 e_2 \cdots i_s e_s i_{s+1}$ in $H$, where $i = i_1, j = i_{s+1}$. From the definition of $D_{\mathcal{A}_H}$, there exists a directed arc $(i_k, i_{k+1})$from $i_k$ to $i_{k+1}$ for $k\in [s]$ in $D_{\mathcal{A}_H}$, which means that there is a directed path from $i$ to $j$ in $D_{\mathcal{A}_H}$. Then $D_{\mathcal{A}_H}$ is strongly connected.
	
	When $D_{\mathcal{A}_H}$ is strongly connected. For $i, j \in V(D_{\mathcal{A}_H})$, $i \neq j$, there exists a directed path $\widehat{P}_s = (i_1, i_2)(i_2, i_3)\cdots(i_s, i_{s+1})$ in $D_{\mathcal{A}_H}$, where $i_1=i,i_{s+1}= j$. From the definition of $D_{\mathcal{A}_H}$, we know there exists  some hyperedges $e_1,e_2,\cdots,e_s$ in $H$
	that $\{i_k, i_{k+1}\}\subseteq e_k$ for $k\in [s]$ in $H$, which means that there is a directed path from  $i$ to $j$ in $H$. Then $H$ is connected.
\end{proof}

\begin{thm}\label{dingli}
	Let $H$ be a connected hypergraph with $n$ vertices, $rank(H)=m$ and $s(H)=s$.
	Then for $\lambda=\rho (\mathcal{A}_H)$, the Equation \eqref{2} has a unique positive solution  $\boldsymbol  {y}$ (up to its scalar multiples).
\end{thm}

\begin{proof}
	Let $\mathcal{A}_{H}$ be the adjacency tensor of $H$, we have
\begin{align*}	
	\left(\mathcal{A}_{H}  \boldsymbol {y}^{s}\right)_{i}
&= \sum_{i_{2},\ldots,i_{s+1}=1}^{n} a_{ii_{2}\cdots i_{s+1}} y_{i_{2}} \cdots y_{i_{s+1}} \\
&=\sum_{k=2}^{m}\sum_{{e}\in E_k(i)} \frac{[(\frac{s}{k-1})!]^{k-1}}{s!}\frac{s!}{[(\frac{s}{k-1})!]^{k-1}} (  \boldsymbol {y}^{e\setminus \{i\}})^{\frac{s}{k-1}}\\
&=\sum_{k=2}^{m}\sum_{{e}\in E_k(i)} (
\boldsymbol {y}^{e\setminus \{i\}})^{\frac{s}{k-1}}.
\end{align*}
Thus, for $ i \in[n]$, Equation \eqref{2} can be expresssed as

\begin{equation}\label{3}
\lambda y_i^{s} = \left(\mathcal{A}_{H}  \boldsymbol{y}^{s}\right)_{i} ,
\end{equation}
i.e.
\begin{align}\label{4}
\mathcal{A}_{H} \boldsymbol {y} ^{s}= \lambda \boldsymbol {y} ^{[s]},
\end{align}
where $\boldsymbol {y}^{[s]}=(y_1^{s},y_2^{s},{\cdots},y_n^{s})^{\mathsf{T}}$.

%
By Theorem \ref{liantong}, we have the tensor $\mathcal{A}_H$ is weakly irreducible.
And by Lemma \ref{perron}, the Equation (\ref{4}) has a unique positive solution $\boldsymbol {y}$ (up to its scalar multiples) corresponding to $\lambda=\rho (\mathcal{A}_H)$.
\end{proof}

\begin{defi}
For a connected hypergraph $H$ with $s(H)=s$,  the unique positive solution $\boldsymbol{y}$ of
$\mathcal{A}_{H}\boldsymbol{y}^{s}=\rho(\mathcal{A}_{H})\boldsymbol{y}^{[s]}$
is called the eigenvector centrality of $H$, denoted as HEC $(\|\boldsymbol  {y}\|_{s+1}=1)$.
\end{defi}

\begin{remark}
Given that $\boldsymbol {x}=\boldsymbol {y}^{[s]}$, and $\boldsymbol {x}, \boldsymbol{y}$ are positive, the rankings of vertices in hypergraphs induced by $\boldsymbol {x}$ and $ \boldsymbol{y}$ are identical.
Hence, we adopt $ \boldsymbol{y}$ as the eigenvector centrality of hypergraphs.
Noting that $ \boldsymbol{y}$ satisfies Equation \eqref{2}, while $\boldsymbol {x}$ solves Equation \eqref{1}.
\end{remark}

\section{Numerical analysis}
\subsection{Algorithm for computing HEC}
Based on the theoretical framework established in Section 3, we now present Algorithm \ref{A1} for computing the proposed eigenvector centrality (HEC) for connected hypergraphs.
Our algorithm employs the ZQW-algorithm to compute tensor eigenvalues and eigenvectors \cite{zhou2013efficient}.

\begin{algorithm}[h]
    \small
    \renewcommand{\algorithmicrequire}{\textbf{Input:}}
    \renewcommand{\algorithmicensure}{\textbf{Output:}}
    \caption{The eigenvector centrality for  hypergraphs (HEC)}
    \label{A1}
    \begin{algorithmic}[1]
        \Require A connected hypergraph $H$, 
        and cardinality set $\{l_1,l_2,\cdots,l_{r}\}$.
        \Ensure The eigenvector centrality for  hypergraph $\boldsymbol {y}$.

        \State Compute $s=s(H)=lcm(l_1-1,l_2-1,\cdots,l_{r}-1)$.
         \State  Construct the adjacency tensor $\mathcal{A}_H$ of order $s+1$.
          \State Choose an $n$-dimension vector ${\boldsymbol {y}}^{(0)} > 0$.
 Let $\mathcal{B}_{H} = \mathcal{A}_{H} + \mathcal{I}$, where $\mathcal{I}$ is the identity tensor.
 Calculate $\boldsymbol {x}^{(0)} =\mathcal{B}_{H}{(\boldsymbol {y}^{(0)})^{s}}$.
 Set $t=1$.
        \State Compute
        \Statex   $\boldsymbol {y}^{(t)}= \frac{{\boldsymbol {x}^{(t-1)}}^{[\frac{1}{s}]}   }
        { ||{\boldsymbol {x}^{(t-1)}}^{[\frac{1}{s}]}||} ,
        {\boldsymbol {x}}^{(t)}={\mathcal{B}}_{H}{({\boldsymbol {y}}^{(t)})^{s}}$,
        \Statex  $\underline{\lambda}_t =\min_{y_i^{(t)}> 0 } \frac{x_i^{(t)}}{(y_i^{(t)})^{s}}  $ ,
                     $\overline{\lambda}_t = \max_{y_i^{(t)}> 0 } \frac{x_i^{(t)}}{(y_i^{(t)})^{s}}$.
        \State If $\underline{\lambda}_t = \overline{\lambda}_t$,
        the algorithm stops and then $\boldsymbol {y}$ is the eigenvector centrality for hypergraph $H$;
        Otherwise, replace $t$ by $t+1$ and go to step 4.
 \end{algorithmic}
\end{algorithm}
\subsection{Baselines}
\textbf{Degree Centrality (DC)} measures the importance of a vertex by the number of neighbors, where a vertex is defined as a neighbor of another vertex if they are contained in a same hyperedge \cite{liu2024survey}. The DC of vertex $i$ is given by
\begin{equation}
	k_i = |N(i)|,
	\label{}
\end{equation}
where $N(i)$ is the set containing all neighbors of vertex $i$ .

\textbf{Hyper Degree Centrality (HDC)} measures the importance of a vertex by counting the number of hyperedges containing the vertex, which has a greater chance of spreading the influence widely \cite{liu2024survey}. The HDC of a vertex $i$ is given by
\begin{equation}
	k_i^H = |E(i)|,
	\label{}
\end{equation}
where $E(i)$ is the hyperedge set containing vertex $i$.

\textbf{Vector Centrality (VC)} evaluates the centrality of vertices by first characterizing the centrality of the hyperedges \cite{kovalenko2022vector}. Specifically, the hypergraph is projected into a line graph. Then the eigenvector centrality values of each hyperedge is computed and evenly distributed to each vertex within the hyperedge. The VC of a vertex $i$ is given by
\begin{equation}
	c_i = \sum_{e\in E(i)}c_{i,e},
	\label{}
\end{equation}
where $E(i)$ is the hyperedge set containing vertex $i$ and $c_{i,e}$ is the centrality value of vertex $i$ distributed from hyperedge $e$.

\textbf{Clique-expansion Centrality (CC)} evaluates the centrality of vertices by computing the eigenvector centrality of the clique expansion graph of the hypergraph. The clique expansion graph of a hypergraph is formed by translating each hyperedge as a complete subgraph (clique) \cite{agarwal2006higher,bretto2013hypergraph,zien2002multilevel}. The HDC of all vertices are solved by equations:
\begin{equation}
 \boldsymbol{x} = \frac{1}{\rho(A)}A\boldsymbol{x},
	\label{}
\end{equation}
where $A$ is the adjacency matrix of the clique expansion graph of the hypergraph.

\subsection{Datasets}
In this section, we introduce the hypergraph datasets used in the subsequent experiments, which are empirical data from multiple domains.

The Email-Enron dataset \cite{Benson-2018-simplicial} has
$143$ vertices and $1542$ hyperedges, where  vertices are  email users and a hyperedge is composed by all users who are included in the same email.

 The Restaurant dataset \cite{amburg2022diverse} has
 $565$ vertices and $601$ hyperedges, where vertices are Yelp users and a hyperedge is composed by users reviewing on the same restaurant.

 The Geometry dataset \cite{amburg2022diverse} has
 $580$ vertices and $1193$ hyperedges, where vertices are MathOverflow users and a hyperedge is  a group of users who answered the same questions related to geometry.

 The Roget dataset \cite{Batagelj2004}, obtained from the Pajek dataset, has
 $1010$ vertices and $997$ hyperedges. vertices in which correspond to different categories in Peter Mark Roget's 1879 edition of the English Thesaurus, while hyperedges are cross-referencing relationships between vocabulary in different categories.

 The Music-blues dataset \cite{ni2019justifying} has
 $1106$ vertices and $694$ hyperedges, where vertices are Amazon users and a hyperedge is composed by all users commented on the same type of music-blues.

 The Film-ratings dataset was initially a bipartite graph from the Koblenz Network Collection \cite{kunegis2013konect}, which is transformed into a hypergraph based on the relationships between vertices. It has
 $2064$ vertices and $1399$ hyperedges, where vertices are movies and a hyperedge is composed by all movies rated by a user.

In the subsequent experiments, we only consider the hypergraph with hyperedges composed of $2$, $3$ and $4$ vertices in the following dataset.	

\subsection{Experiments}
This paper use a tensor representation to study hypergraphs, which ensures a bijective mapping between the hypergraph and its tensor representation.
To more clearly demonstrate that the centrality proposed in this paper preserves the high-order structure of hypergraphs while also considering the influence of adjacent vertices, we conduct the following experiments.

Figure \ref{fig1} presents a non-uniform sunflower hypergraph consisting of three hyperedges and its clique expansion graph.
We calculated the scores of vertices in Figure \ref{fig1}, under five different centrality measures (DC,HDC,CC,VC,HEC),
and the results are shown in Table \ref{table1}.

\begin{figure}[H]
\centerline{\includegraphics[scale=0.4]{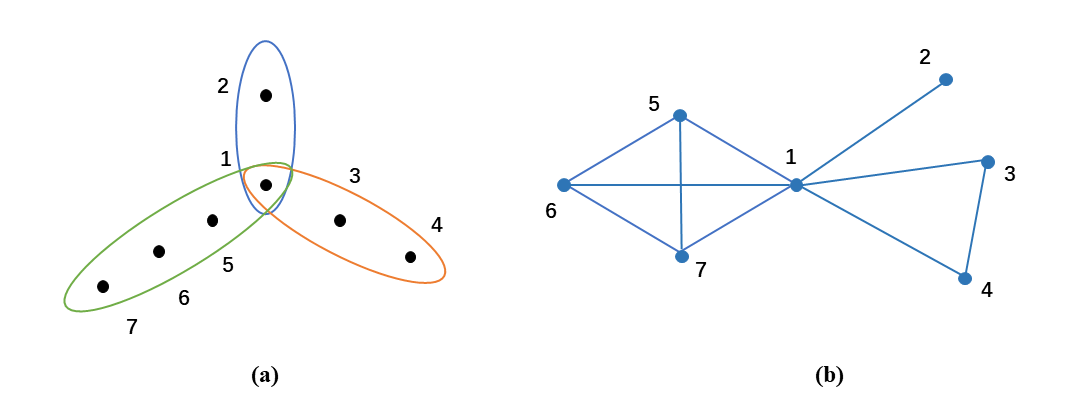}}
\caption{(a) is a non-uniform sunflower hypergraph with three hyperedges and (b) is its clique expansion graph.}
\label{fig1}
\end{figure}

\begin{table*}[!t]
    \tiny
    \centering
    \caption{The centrality scores of vertices in Figure \ref{fig1}.}
    \label{table1}
    \begin{tabular}{lcccccc}  
    \toprule
        \textbf{Vertex}  & \textbf{DC}  & \textbf{HDC} & \textbf{CC}& \textbf{VC} & \textbf{HEC}\\ \midrule
        \textbf{$1$}     & 6	 &  3  &  0.9571 &   0.9993 &   0.8451                 \\
        \textbf{$2$}     & 	1    &  1  &  0.2851 &   0.4612 &     0.7930        \\
        \textbf{$3,4$}  &  2     &  1  &  0.4061 &   0.3075 &      0.7440              \\
         \textbf{$5,6,7$} & 3	 &  1     &  0.7053 &  0.2306  &   0.6981                 \\
    \bottomrule
    \end{tabular}
\end{table*}

From Table \ref{table1}, the hub vertex $1$ attains the highest centrality score under all centrality measures.
Under HDC, all vertices located in the petals of the sunflower receive identical scores.
In contrast, both DC and CC assign higher scores to vertices situated in larger petals, whereas VC and HEC exhibit the opposite trend, favoring vertices in smaller petals with higher centrality values.
Exactly, DC, HDC and CC fail to capture the varying sizes of hyperedges, thereby losing critical high-order structural information.

Both HEC and CC incorporate the high-order structure of hypergraphs, while simultaneously exhibiting a preference for assigning higher scores to smaller groups.
Small groups typically imply stronger interactions and more efficient collaboration.
In security or information diffusion networks, such small groups may serve as critical control points or origins of propagation. For example, within an organization, a small decision-making team could exert more practical influence than a large department.

However, a closer inspection reveals that VC assigns excessively high scores to the central vertex, while vertices on other edges receive scores far lower than that of the center.
In contrast, HEC distributes scores more equitably across vertices, as it incorporates not only the high-order structure of hypergraphs but also the mutual interactions among vertices.
Consequently, HEC emerges as an effective tool for analyzing vertex importance in non-uniform hypergraphs.

Below, we conduct a systematic experimental analysis and comparison of five hypergraph centrality methods (DC, HDC, CC, VC, HEC) on the six real-world hypergraphS mentioned above. Specifically, we compute the scores of the five centrality measures on each hypergraph and plot the corresponding matrix scatter plots to visually illustrate the correlations and distributional differences among the different centrality measures.

\begin{figure}[!t]
  \centering
   \subfloat[Email-Enron.]{%
    \begin{minipage}[t]{0.44\textwidth}
      \centering
      \includegraphics[width=\linewidth]{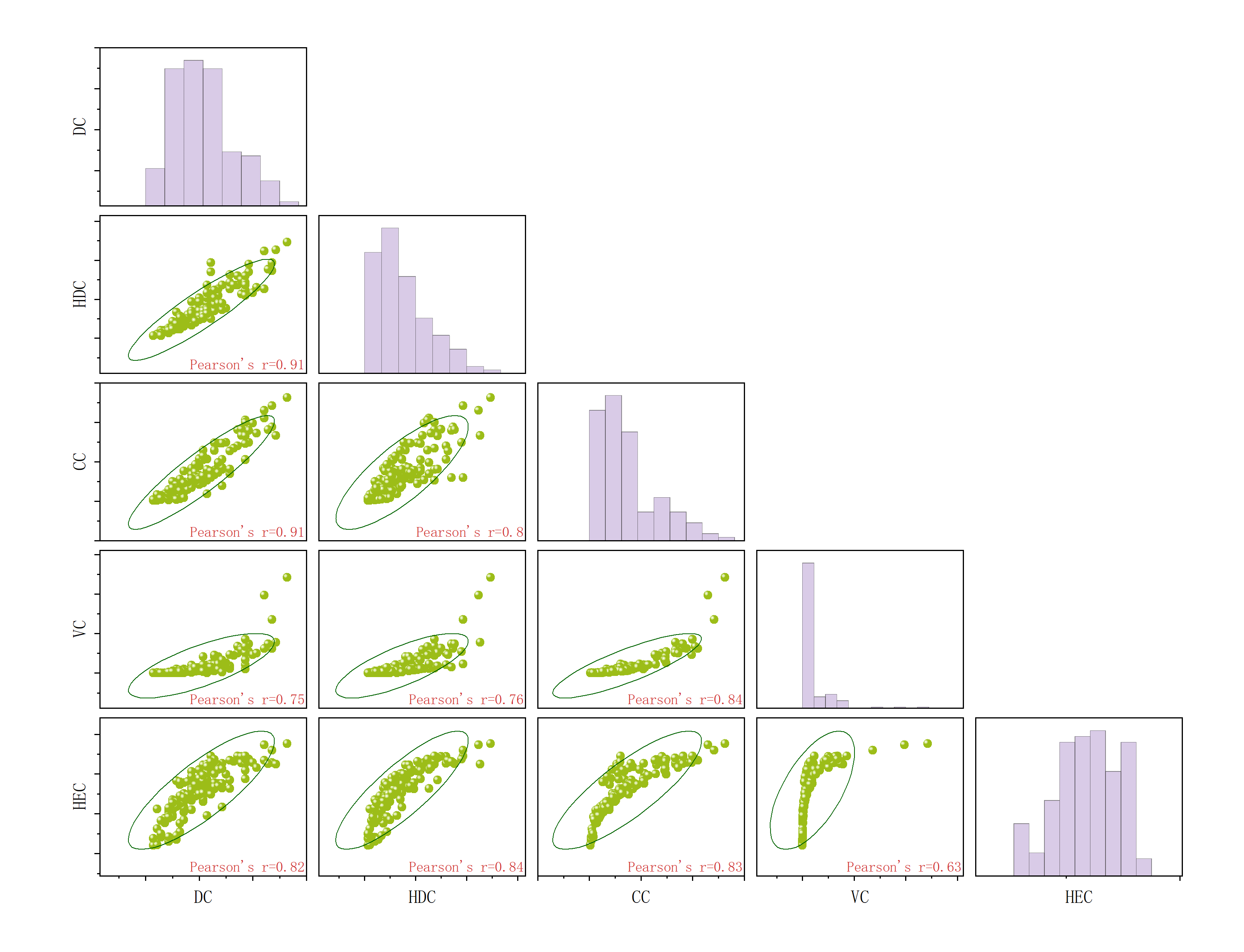}
    \end{minipage}%
  }
  \subfloat[Restaurant.]{%
    \begin{minipage}[t]{0.44\textwidth}
      \centering
      \includegraphics[width=\linewidth]{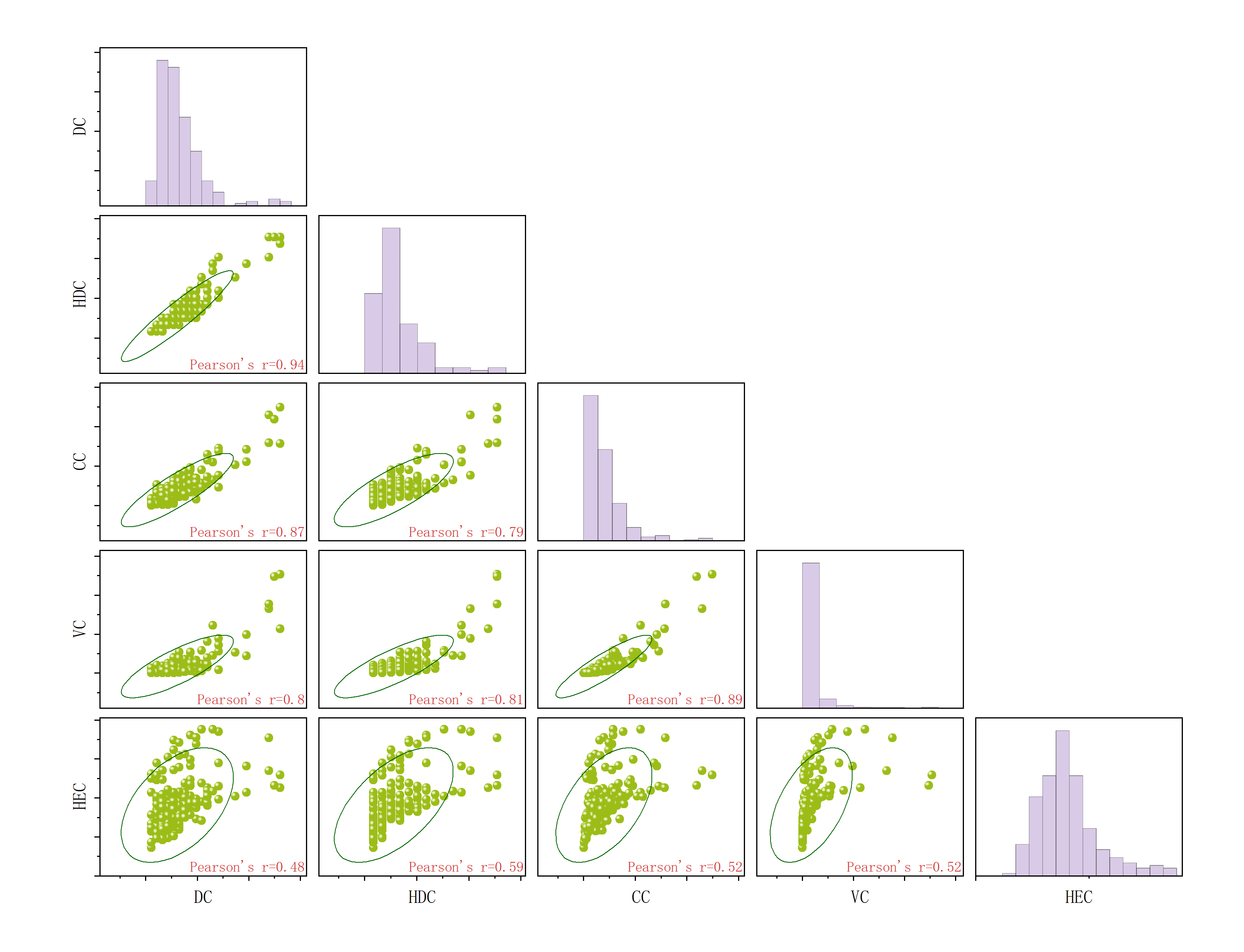}
    \end{minipage}%
  }
    \hfill
  \subfloat[Geometry.]{%
    \begin{minipage}[t]{0.44\textwidth}
      \centering
      \includegraphics[width=\linewidth]{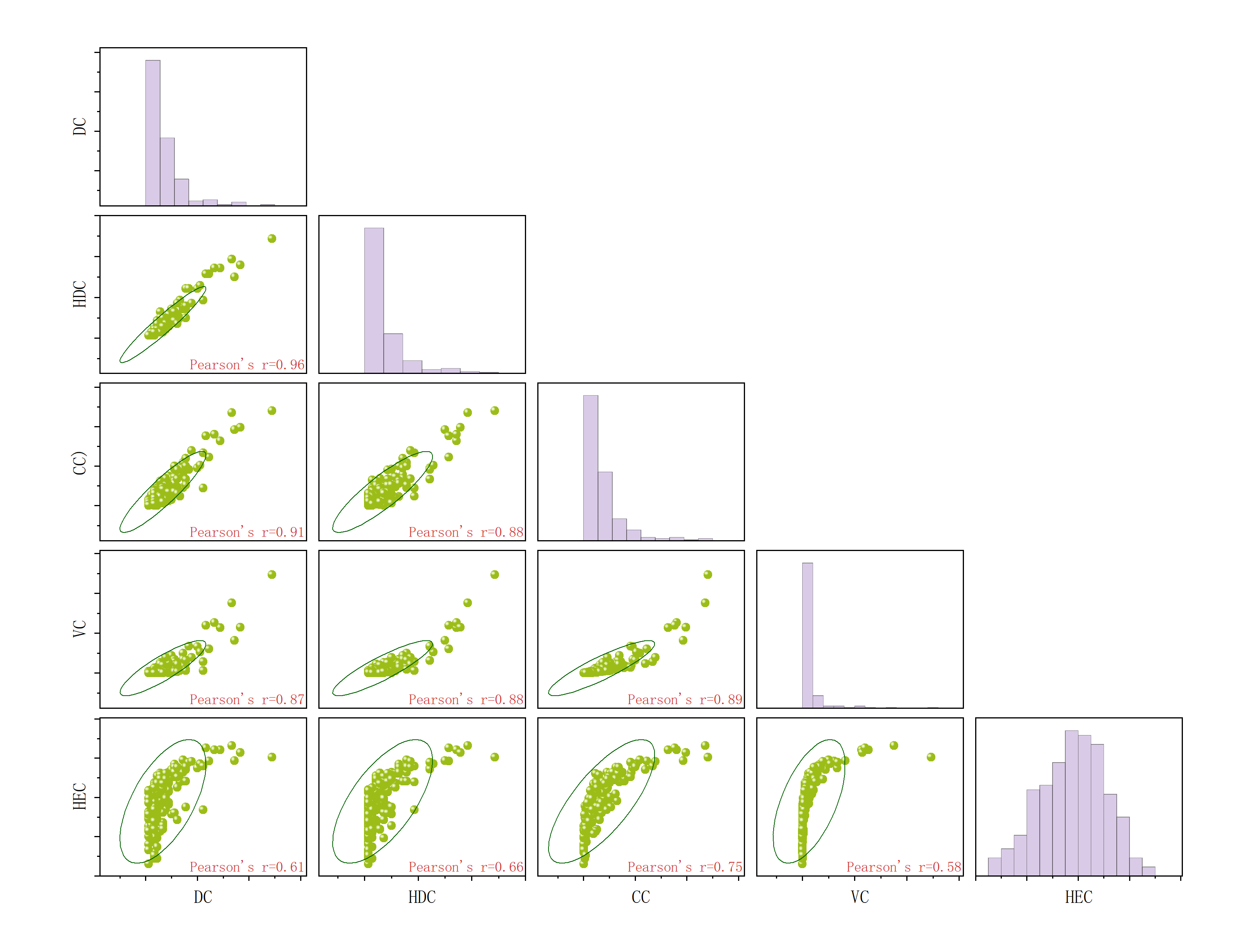}
    \end{minipage}%
  }
    \subfloat[Roget.]{%
    \begin{minipage}[t]{0.44\textwidth}
      \centering
      \includegraphics[width=\linewidth]{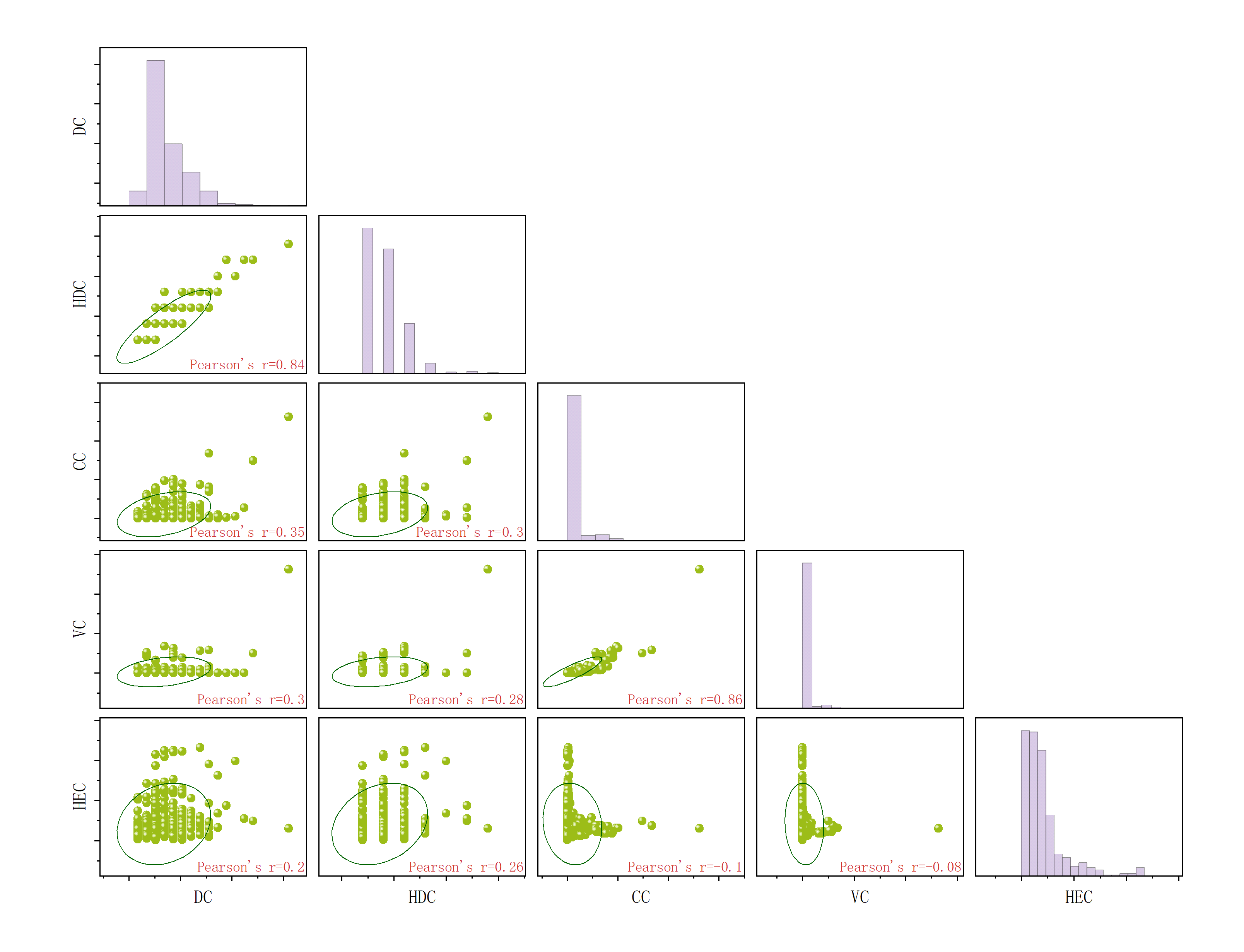}
    \end{minipage}%
  }
    \hfill
  \subfloat[Music-blues.]{%
    \begin{minipage}[t]{0.44\textwidth}
      \centering
      \includegraphics[width=\linewidth]{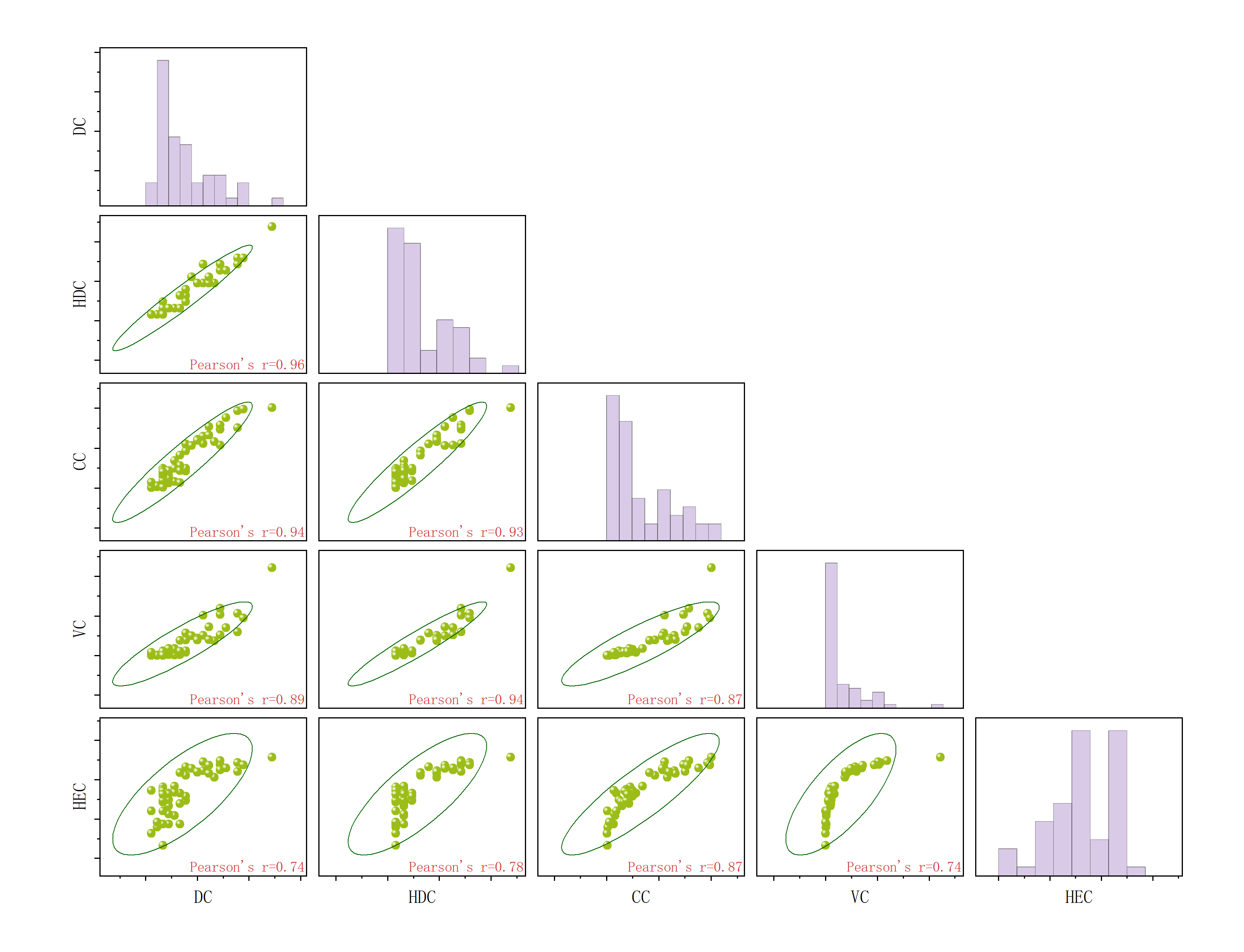}
    \end{minipage}%
  }
    \subfloat[Film-ratings.]{%
    \begin{minipage}[t]{0.44\textwidth}
      \centering
      \includegraphics[width=\linewidth]{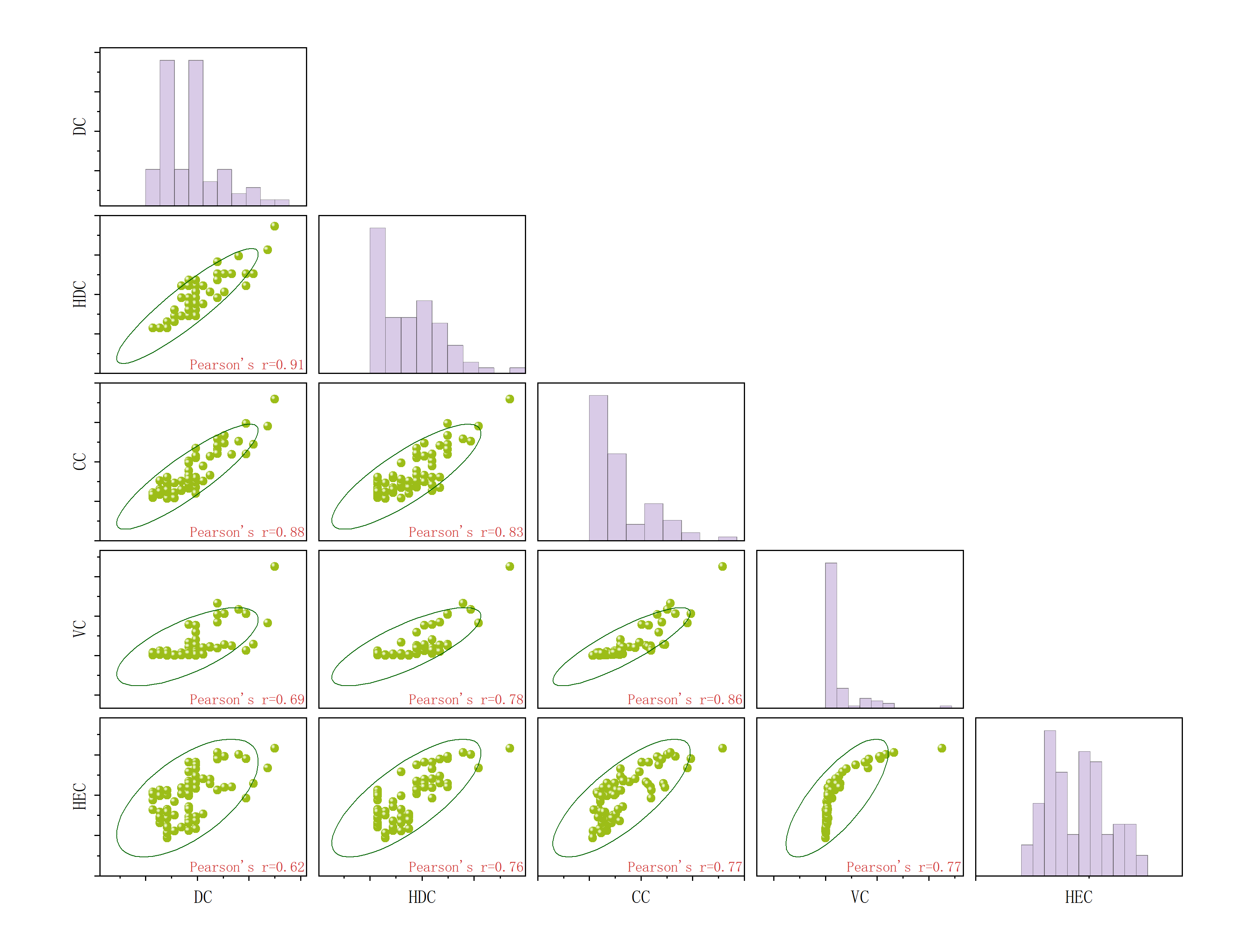}
    \end{minipage}%
  }
  \caption{The matrix scatter plots for the six hypergraphs under the five centrality measures, with correlations quantified by Pearson correlation coefficients.}
  \label{fig2}
\end{figure}

From Figure \ref{fig2}, we have HEC exhibits moderate to strong correlations with the four classical centrality measures across different networks, indicating that the information captured by HEC shares some commonality with existing centralities but is not entirely consistent.
Among them, the correlation between HEC and VC is the lowest in most hypergraphs (except for Geometry), dropping to just $0.08$ in Roget, suggesting that their linear association is the weakest and reflects different focal points. HEC and HDC maintain moderate to strong correlations across all six hypergraphs ($0.26-0.84$), with particularly notable values in Email-Enron and Music-blues, indicating that both measures rely heavily on a vertex's neighbors.
The correlation between HEC and DC, in contrast, exhibits strong hypergraph dependence, ranging from only $0.20$ in Roget to as high as $0.82$ in Email-Enron.
Notably, Roget stands out as an outlier across all centrality pairs: the correlations between HEC and the four classical centralities are extremely low ($0.08-0.26$), while the classical centralities themselves remain highly correlated with each other (e.g., DC-HDC and CC-VC $\geq$ $0.84$), indicating that HEC provides a completely different ranking of vertex importance in this hypergraph, capturing unique structural information.

Among the classical centrality measures, strong correlations are commonly observed.
DC-HDC and CC-VC exhibit strong correlations across all networks.
In contrast, the correlations of DC-CC and HDC-CC show considerable variation: they reach $0.91, 0.94,$ and $0.87$ in Geometry, Music-blues, and Restaurant, respectively, but drop sharply to $0.35$ and $0.30$ in Roget, indicating that the relationships among some centrality measures are heavily dependent on hypergraph topology.
In summary, as a novel centrality measure, HEC overlaps with classical metrics to some extent in most hypergraphs.
However, its weak correlation with VC and its complete divergence from traditional centralities in hypergraphs such as Roget demonstrate that HEC can provide complementary information beyond what classical measures capture, offering unique structural discriminability in assessing vertex importance in hypergraphs.

In complex networks, robustness is typically defined as the network's ability to maintain its fundamental functions and structural integrity when subjected to random failures or malicious attacks \cite{cicchini2024robustness}.
Network robustness is commonly quantified by the relative size of the largest connected component (LCC). The network is considered more robust if it can maintain a larger LCC throughout an attack process; consequently, a faster decline in the LCC signifies poorer robustness \cite{lou2022learning}.
We next conduct experiments on the six hypergraphs. By removing vertices in descending order according to several centrality rankings, we obtain the LCC decay curves shown in Figure \ref{fig3}, where the curve labeled ``random" corresponds to the baseline of random vertex removal.

\begin{figure}[!t]
  \centering
   \subfloat[Email-Enron.]{%
    \begin{minipage}[t]{0.35\textwidth}
      \centering
      \includegraphics[width=\linewidth]{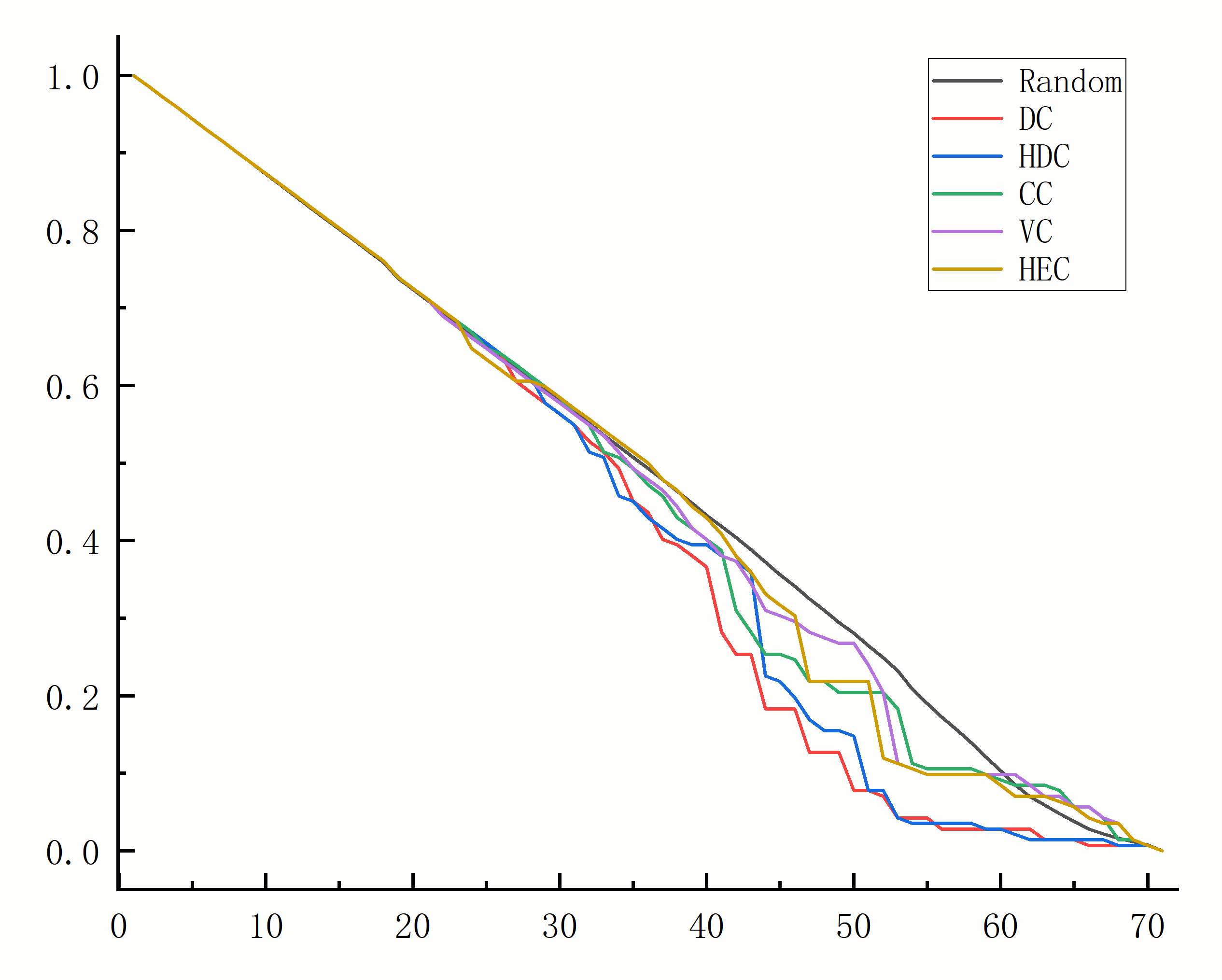}
    \end{minipage}%
  }
  \subfloat[Restaurant.]{%
    \begin{minipage}[t]{0.35\textwidth}
      \centering
      \includegraphics[width=\linewidth]{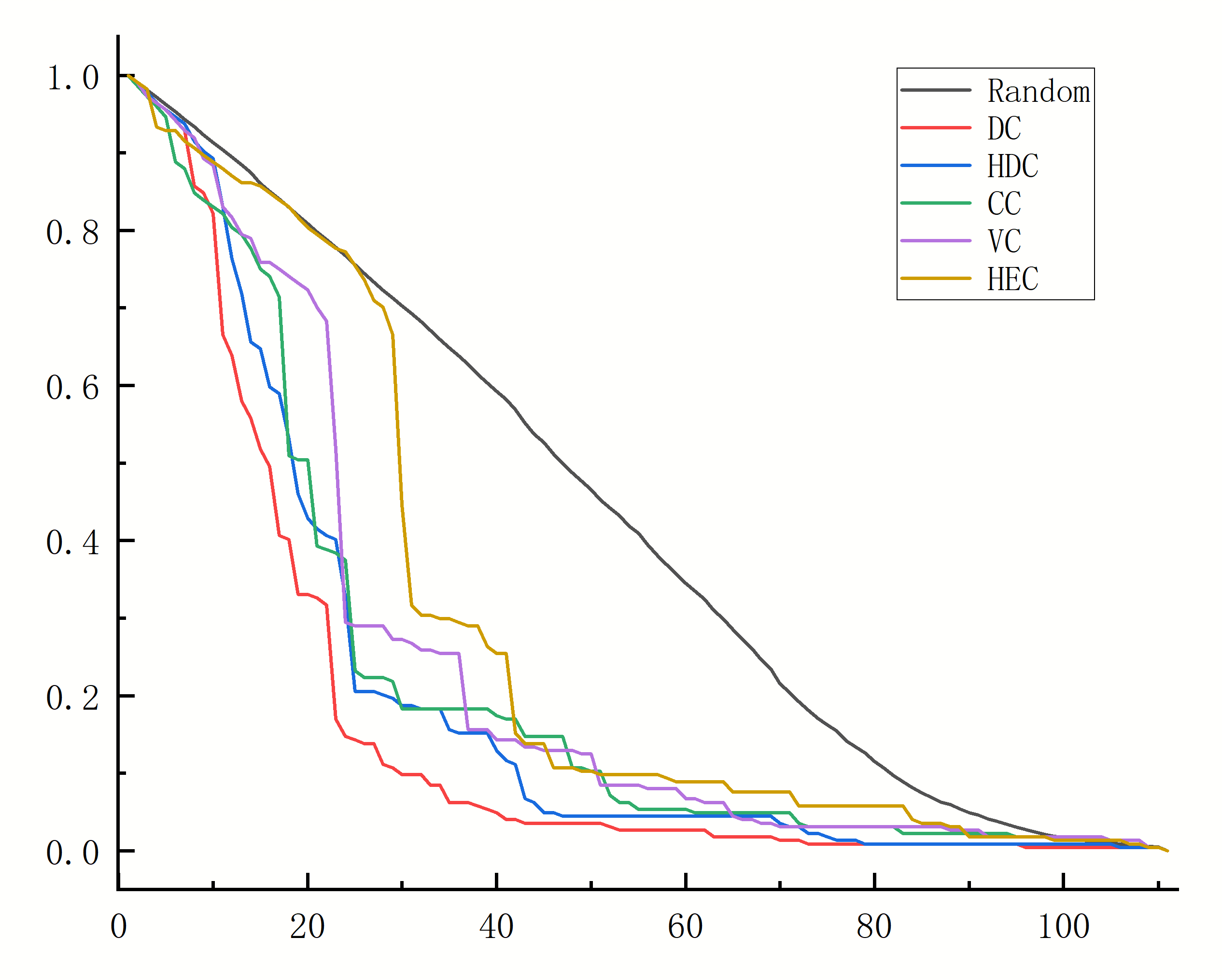}
    \end{minipage}%
  }
  \subfloat[Geometry.]{%
    \begin{minipage}[t]{0.35\textwidth}
      \centering
      \includegraphics[width=\linewidth]{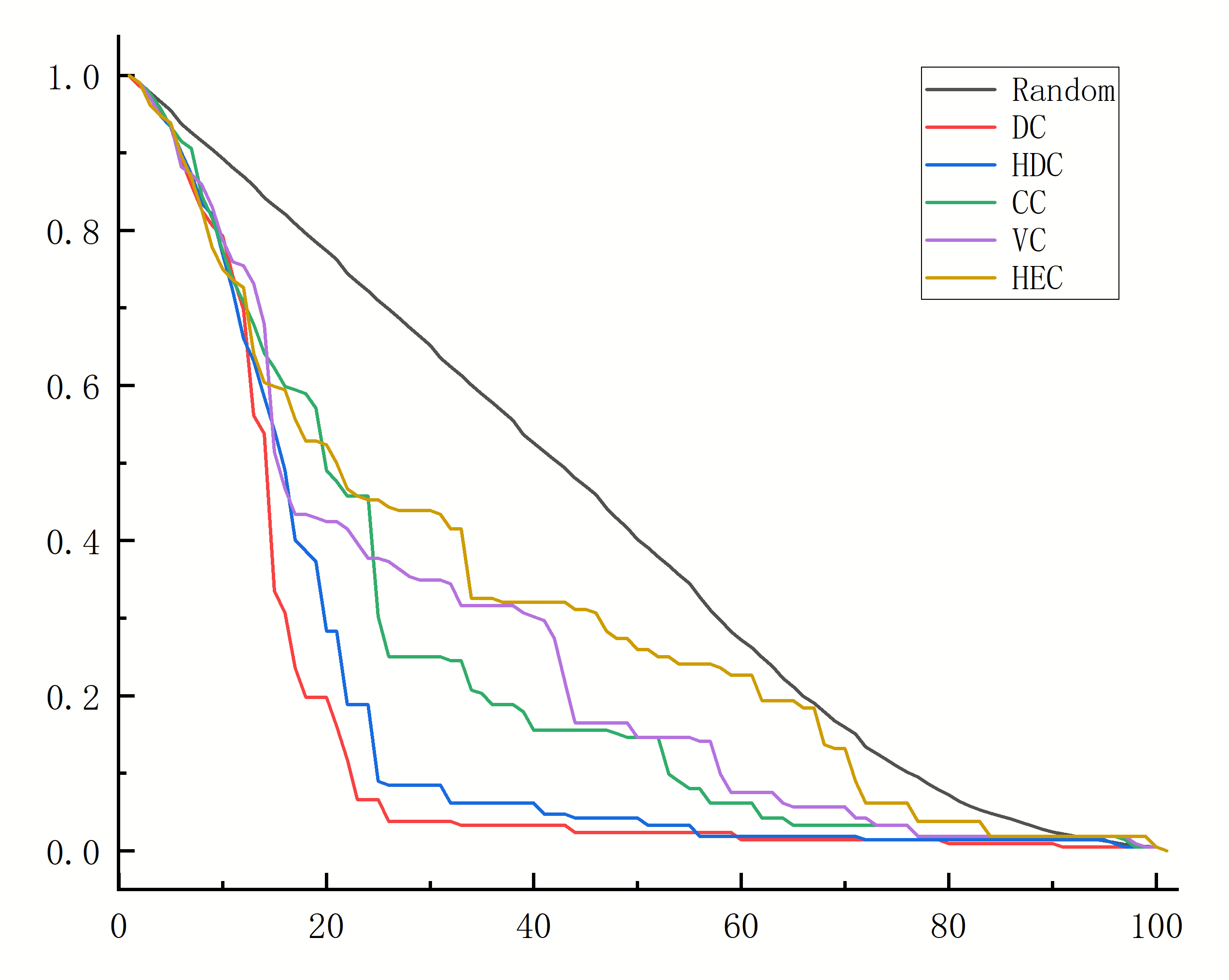}
    \end{minipage}%
  }
    \hfill
    \subfloat[Roget.]{%
    \begin{minipage}[t]{0.35\textwidth}
      \centering
      \includegraphics[width=\linewidth]{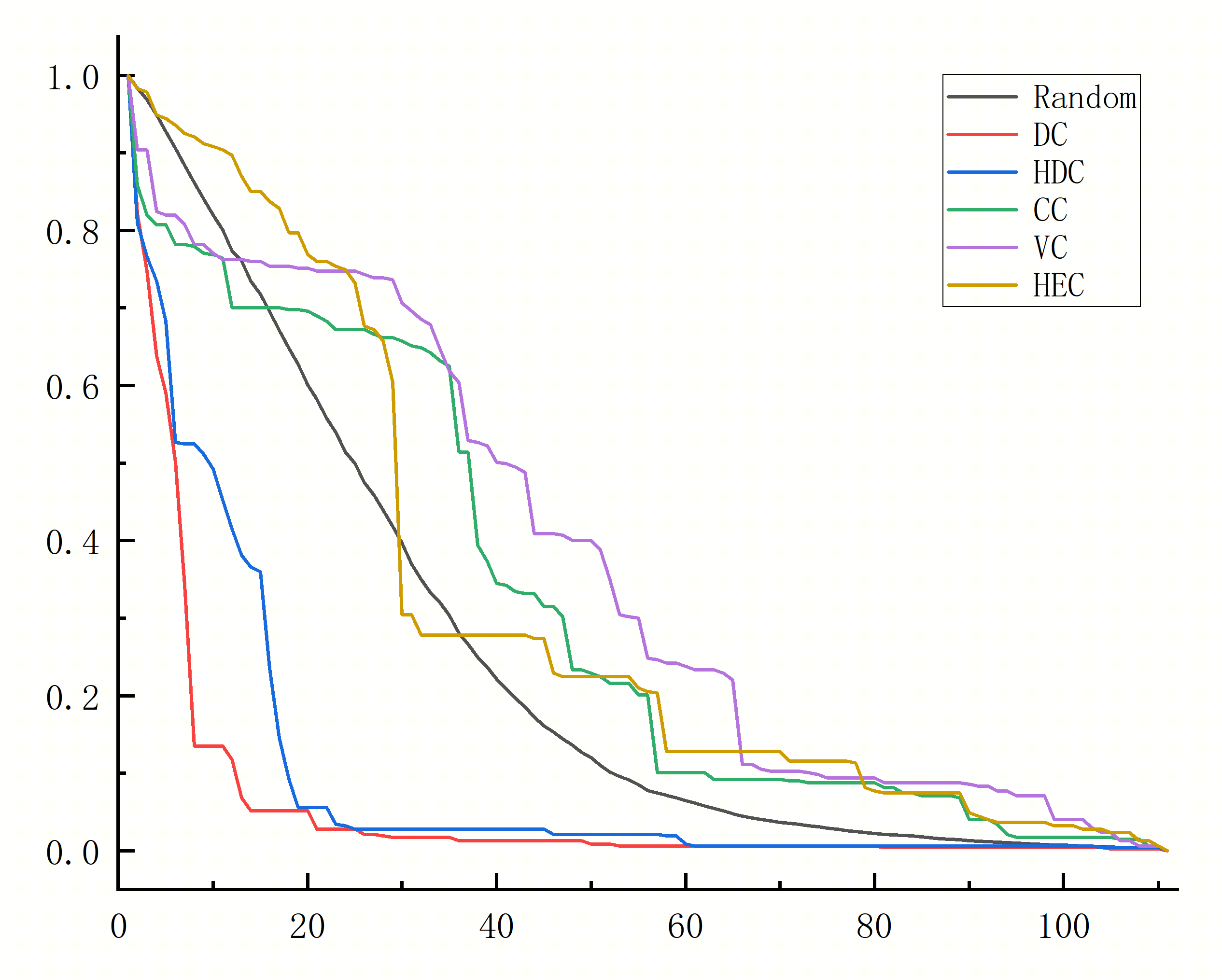}
    \end{minipage}%
  }
  \subfloat[Music-blues.]{%
    \begin{minipage}[t]{0.35\textwidth}
      \centering
      \includegraphics[width=\linewidth]{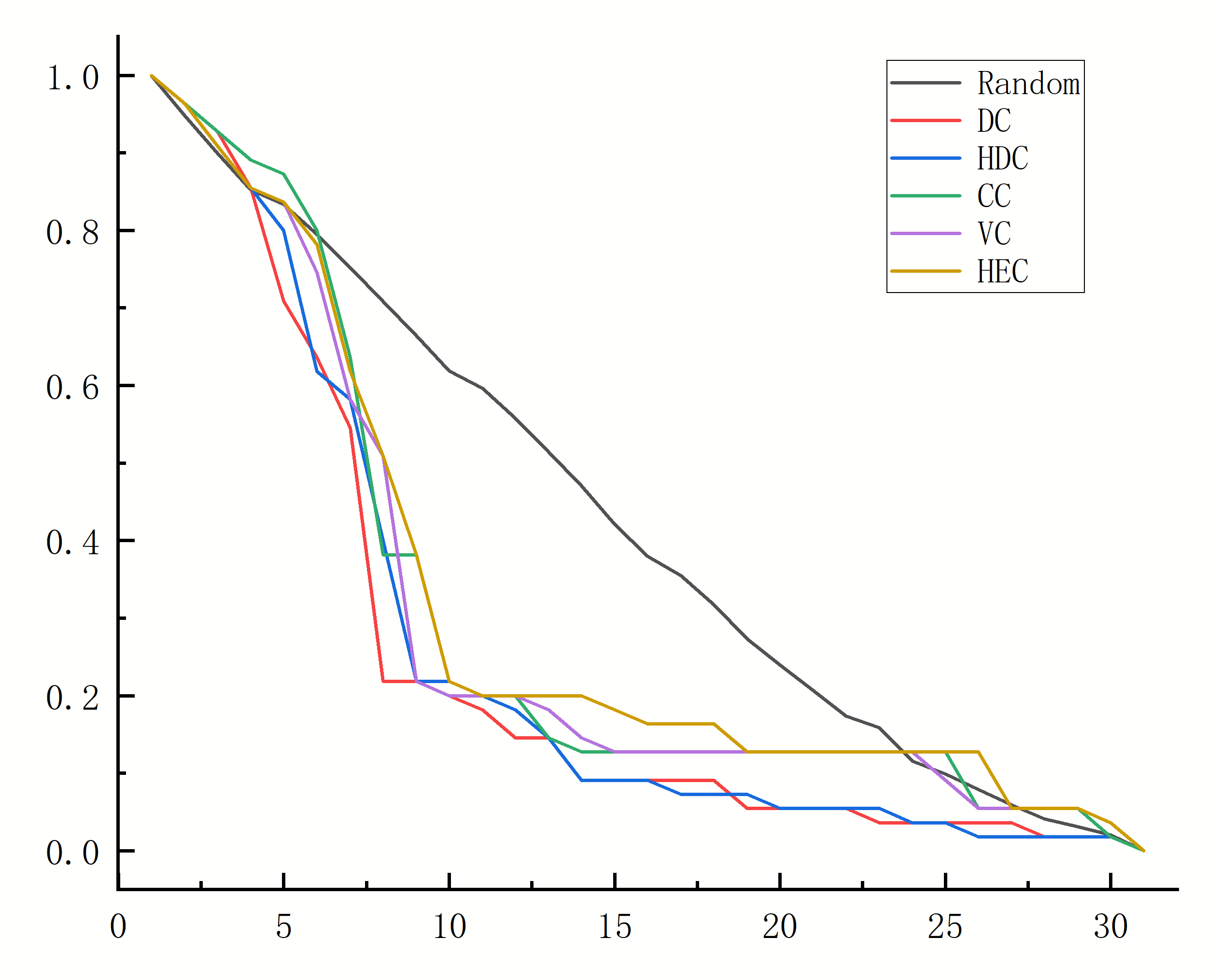}
    \end{minipage}%
  }
    \subfloat[Film-ratings.]{%
    \begin{minipage}[t]{0.35\textwidth}
      \centering
      \includegraphics[width=\linewidth]{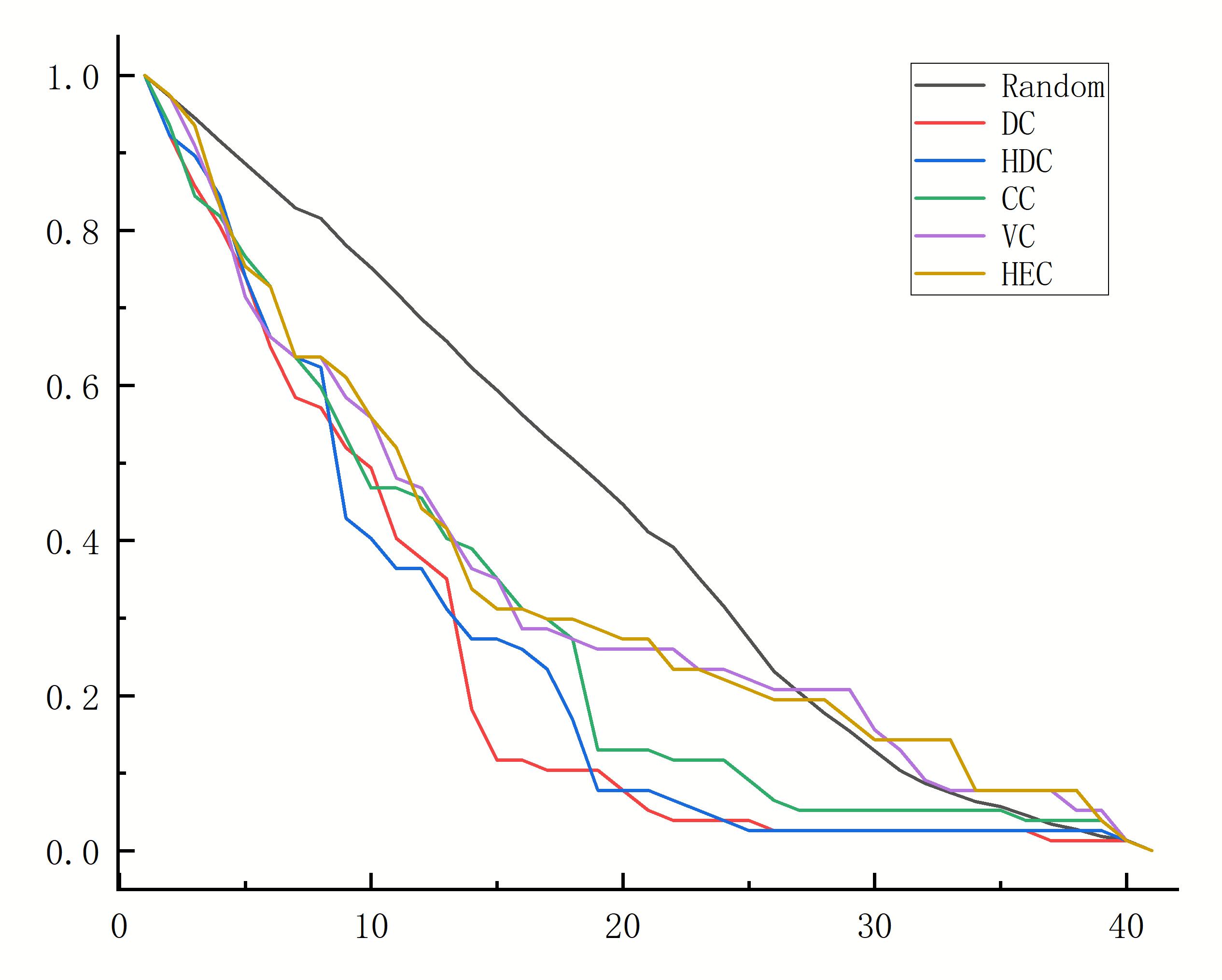}
    \end{minipage}%
  }
  \caption{The LCC decay curves obtained by removing vertices in descending order of five different centrality measures across the six hypergraphs.}
  \label{fig3}
\end{figure}

From Figure \ref{fig3}, we observe that in the Email-Enron, Geometry, Music-blues, and Film-ratings, the HEC curve largely coincides with those of the other centralities during the early stage of attacks. This indicates that HEC is effective at identifying vertices that play critical structural roles in hypergraphs.
In the middle and late stages of attacks, the HEC curve generally lies above the curves of the other centralities and remains close to that of VC. This is because both HEC and VC account for the influence of hyperedge sizes on vertex importance. In particular, HEC incorporates a global recursive mechanism in which a vertex's centrality is influenced by the other vertices sharing the same hyperedges, thereby capturing richer higher-order information. As a result, the LCC decay rate under HEC is slower.

In Roget, DC and HDC exhibit a sharp drop in LCC at the very beginning of the attack, whereas the LCC curves for HEC, CC, and VC even lie above the random removal baseline. This suggests that in hypergraphs with certain special topological structures, relying solely on degree-based measures can lead to drastically different assessments of vertex importance compared to methods that consider global structural information. In summary, HEC not only identifies the most critical vertices in a hypergraph but also prioritizes the overall structural resilience and connectivity maintenance capability, rather than merely local connectivity strength. Thus, HEC provides a valuable theoretical foundation for robustness analysis and critical vertex protection in hypergraphs.

The Jaccard index \cite{jaccard1901etude}, also known as the Jaccard similarity coefficient, is a statistic used for comparing the similarity and diversity of two finite sample sets. It is defined as the size of the intersection of the sets divided by the size of their union.
For two sets $A$ and $B$, the Jaccard index $J(A,B)$is given by:
\begin{align*}
J(A,B)=\frac{|A\cap B|}{|A\cup B|},
\end{align*}
the index ranges from $0$ to $1$.

To quantify the consistency among different centrality measures in identifying the most important vertices, we compute the Jaccard index.
For each hypergraph, under the five centrality rankings, we extract the top $k$ vertices (with $k=5,10,15,20,25$ ) and calculate the Jaccard index for each pair of centrality methods.
The Jaccard index is defined as the size of the intersection divided by the size of the union of the two vertex sets, ranging from $0$ to $1$, where a higher value indicates greater overlap in the top-ranked vertices.
Based on these results, we plot heatmaps for the six datasets to visually illustrate the similarity patterns among different centrality measures at various thresholds $k$.

\begin{figure}[!t]
  \centering
  \subfloat[Email-Enron.]{%
    \begin{minipage}[t]{0.33\textwidth}
      \centering
      \includegraphics[width=\linewidth]{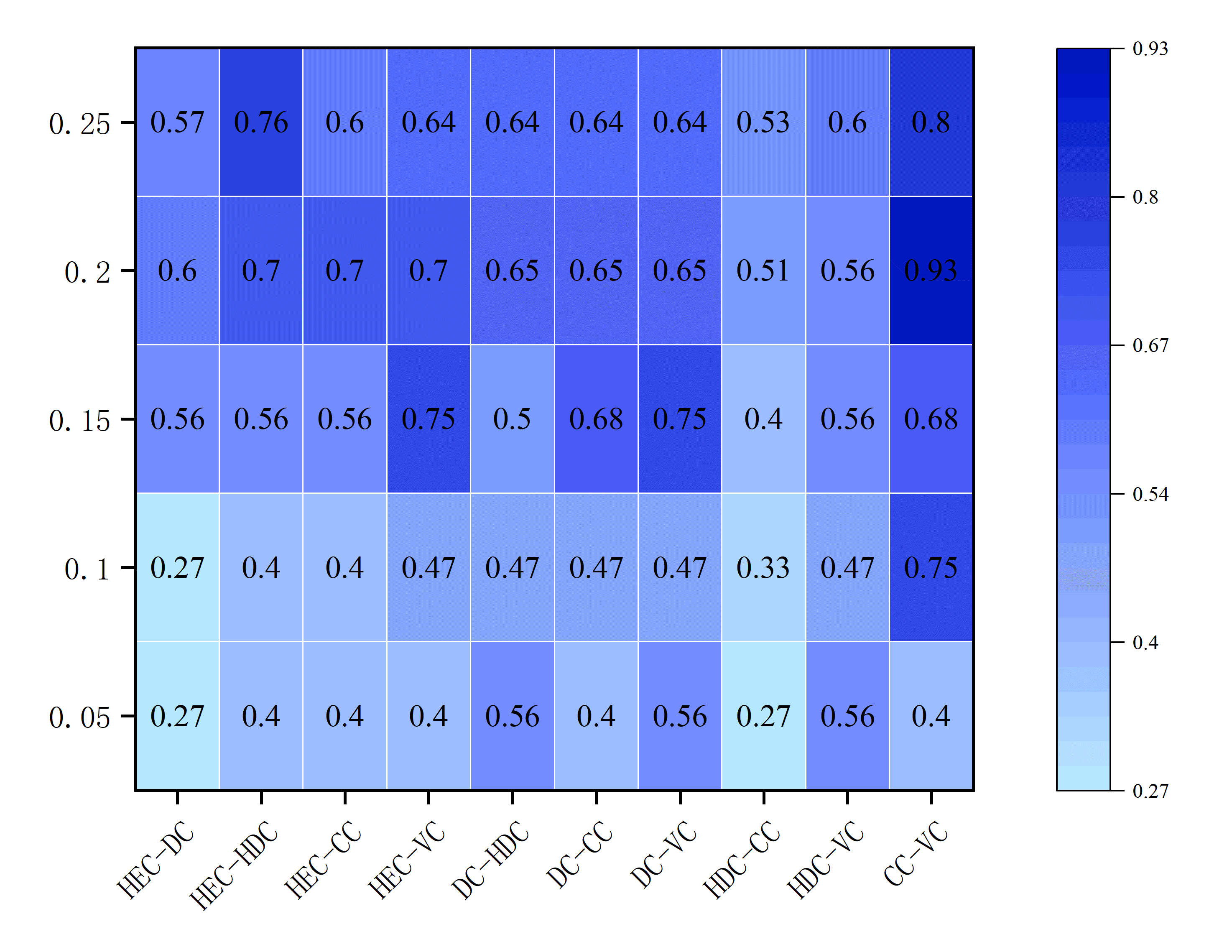}
    \end{minipage}%
  }\hfill
  \subfloat[Restaurant.]{%
    \begin{minipage}[t]{0.33\textwidth}
      \centering
      \includegraphics[width=\linewidth]{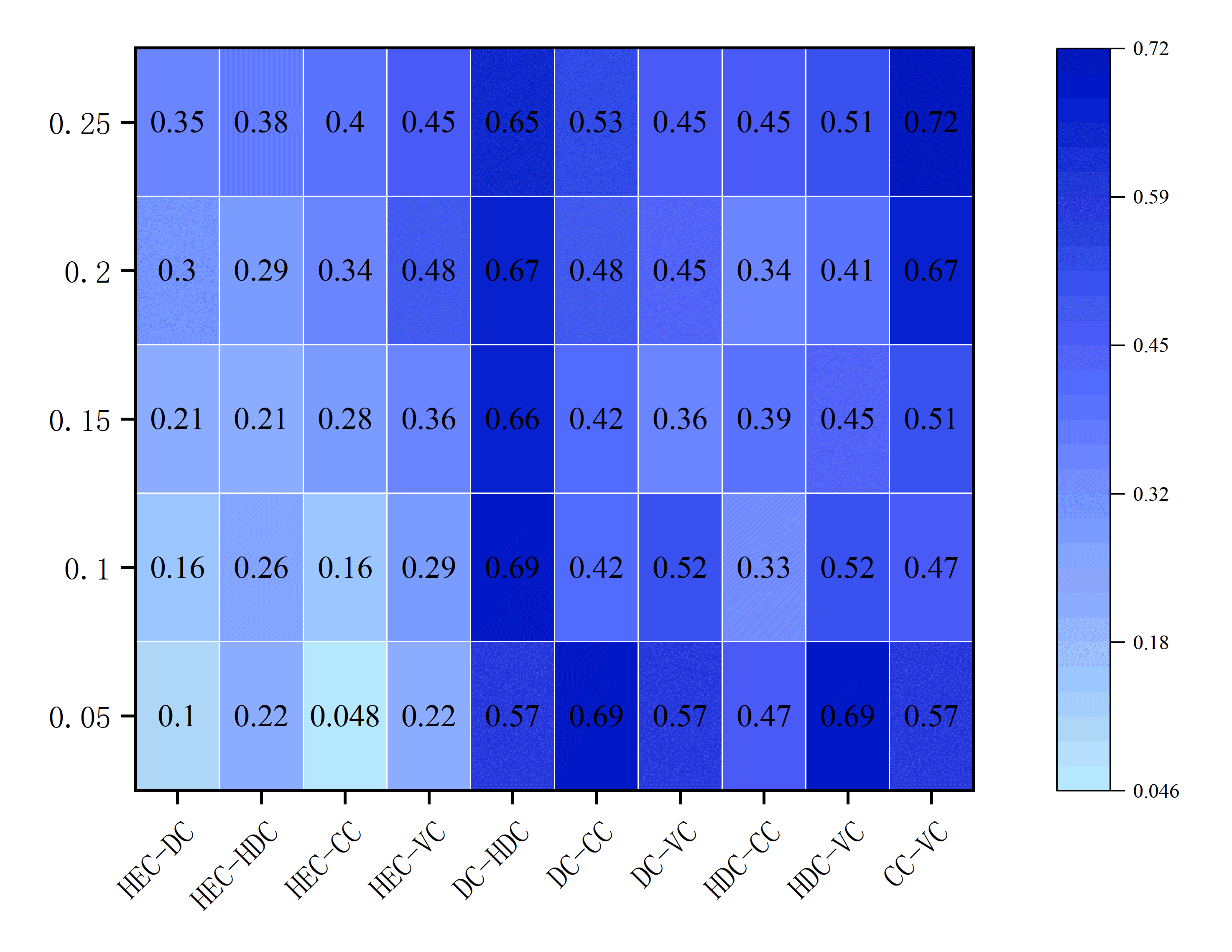}
    \end{minipage}%
  }\hfill
  \subfloat[Geometry.]{%
    \begin{minipage}[t]{0.33\textwidth}
      \centering
      \includegraphics[width=\linewidth]{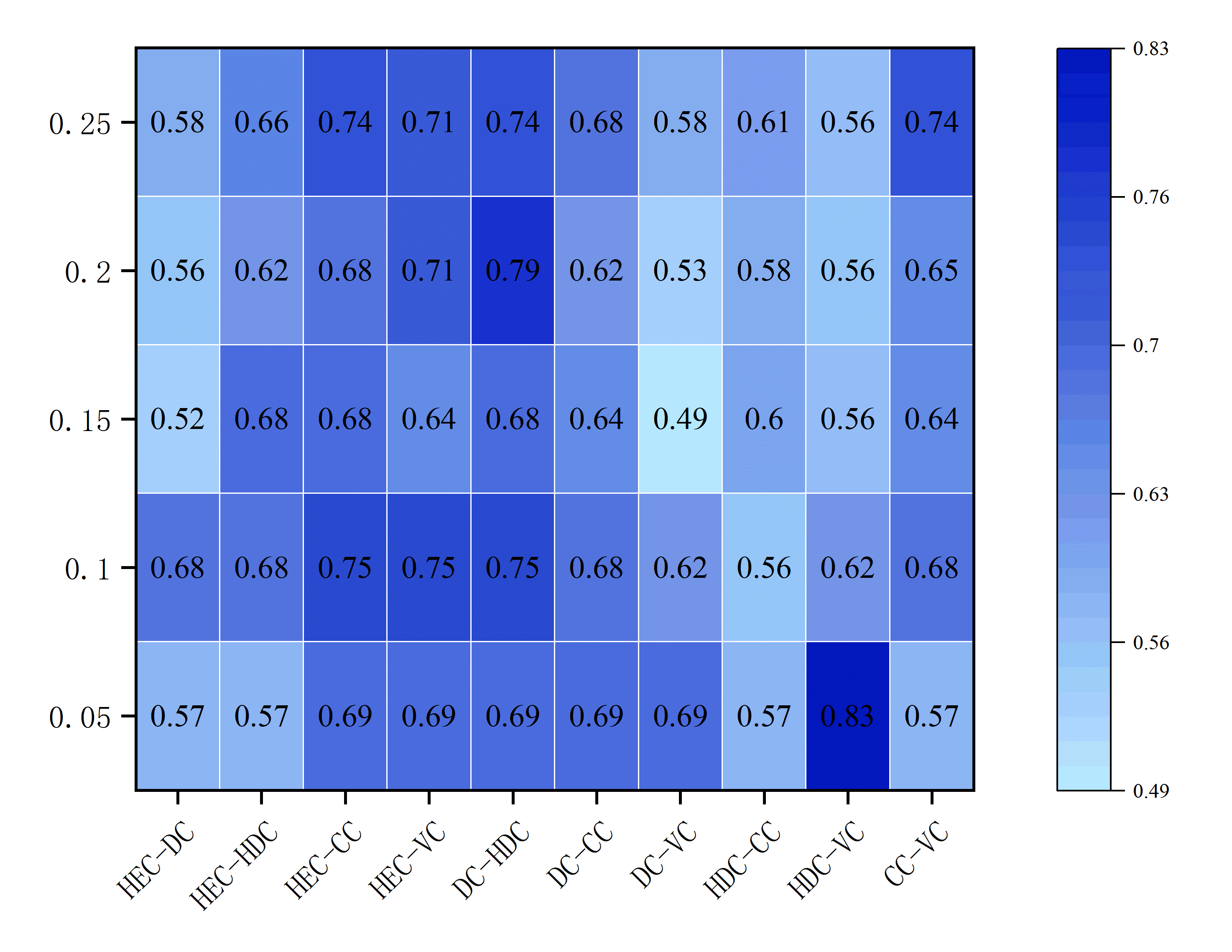}
    \end{minipage}%
  }\\
  \subfloat[Roget.]{%
    \begin{minipage}[t]{0.33\textwidth}
      \centering
      \includegraphics[width=\linewidth]{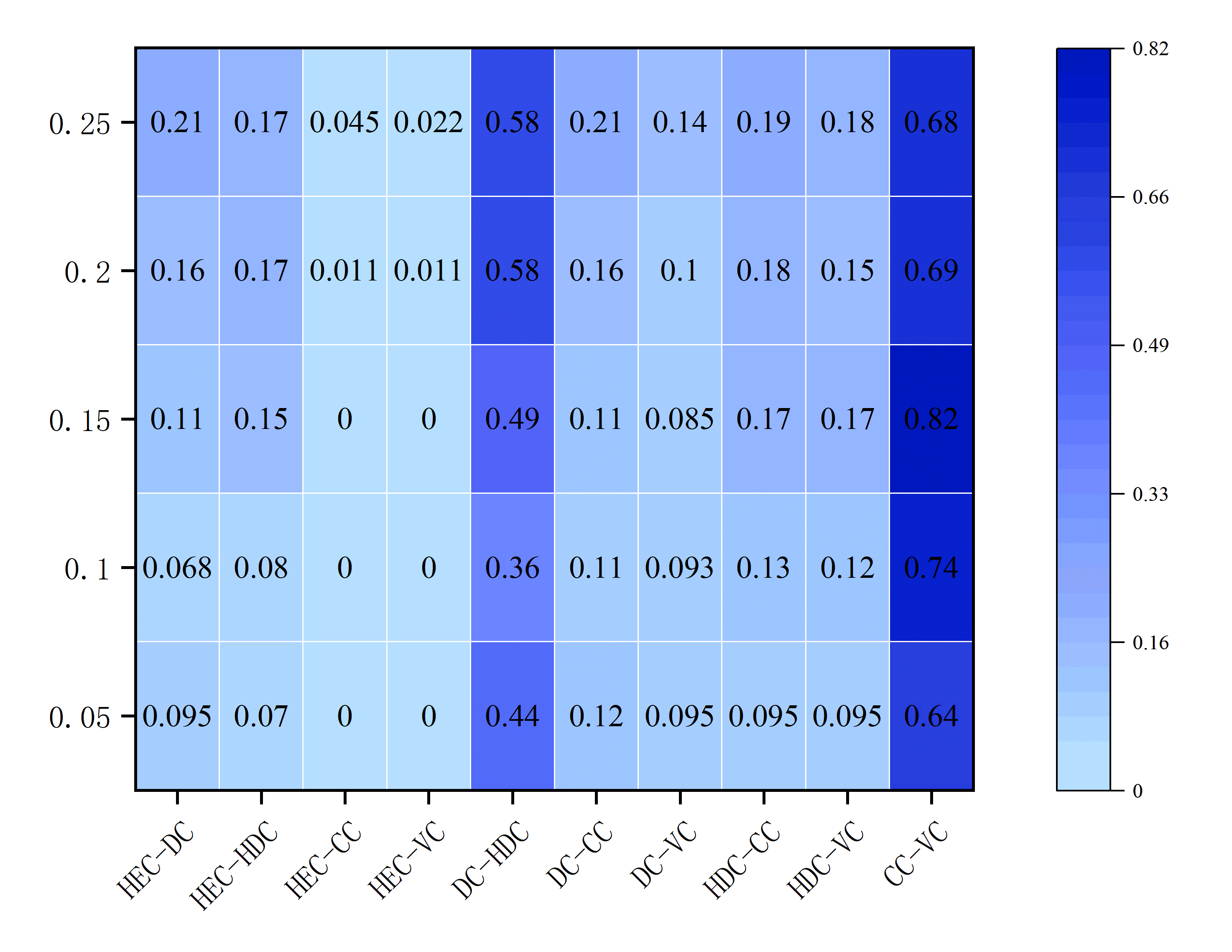}
    \end{minipage}%
  }\hfill
  \subfloat[Music-blues.]{%
    \begin{minipage}[t]{0.33\textwidth}
      \centering
      \includegraphics[width=\linewidth]{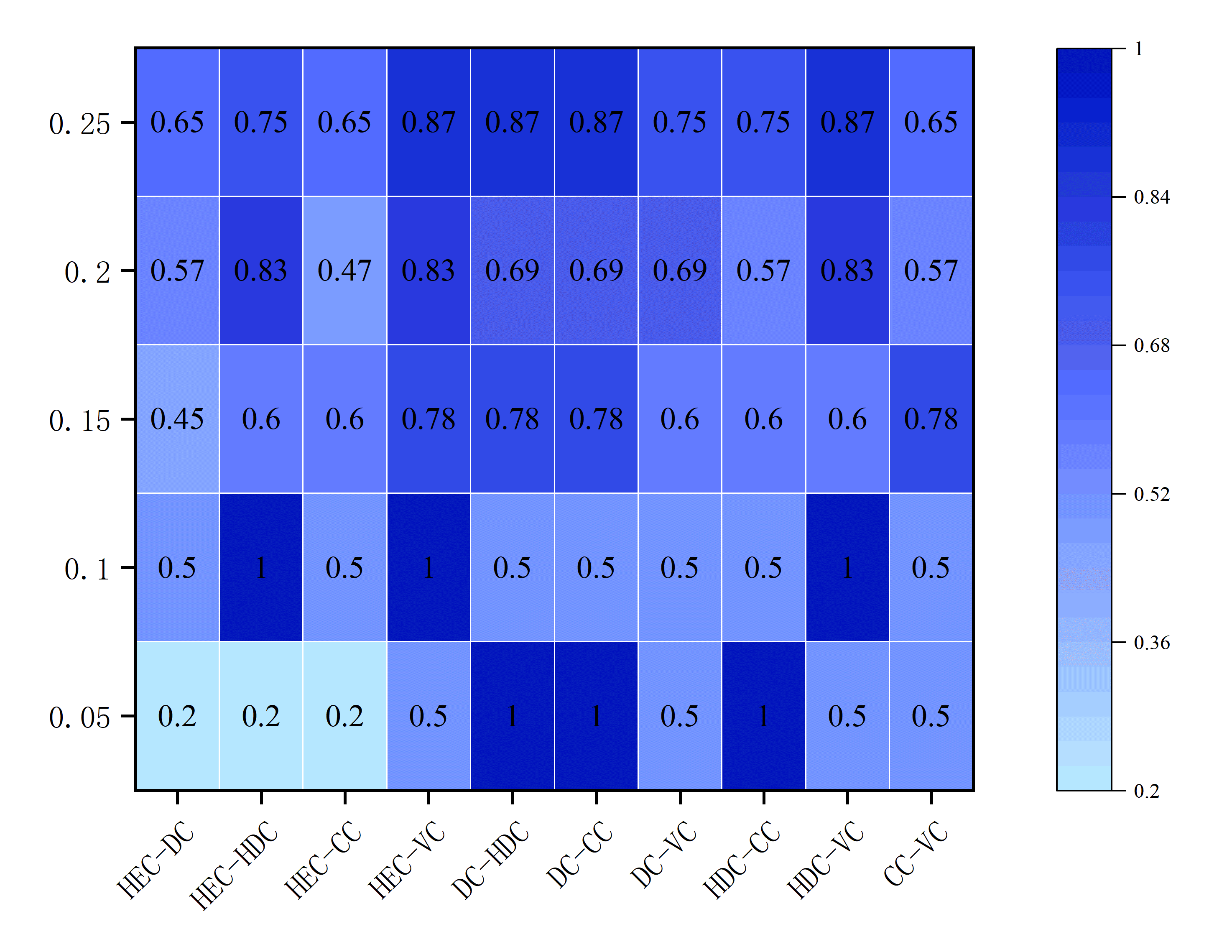}
    \end{minipage}%
  }\hfill
  \subfloat[Film-ratings.]{%
    \begin{minipage}[t]{0.33\textwidth}
      \centering
      \includegraphics[width=\linewidth]{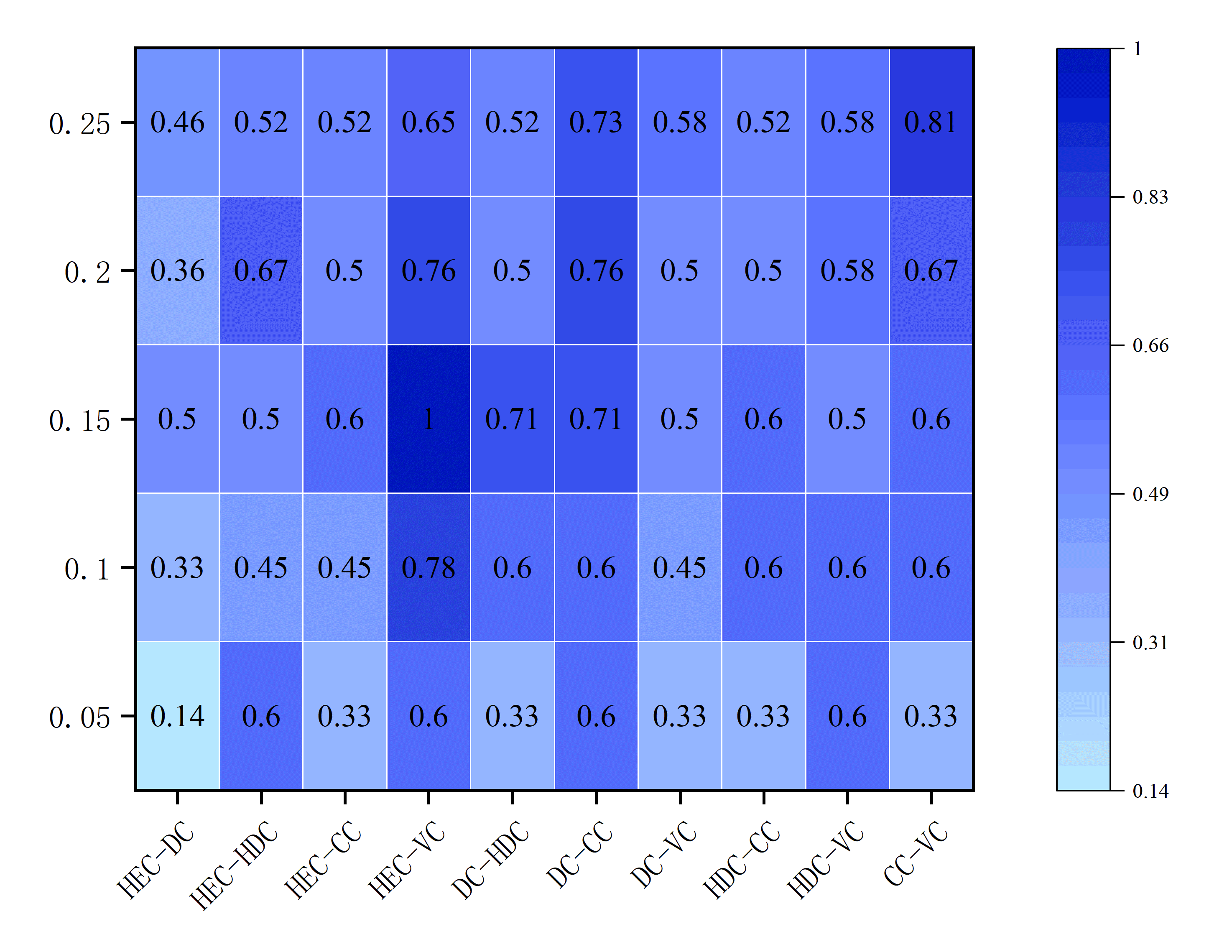}
    \end{minipage}%
  }
  \caption{Heatmaps of the Jaccard index among five centrality measures at different top-$k$ $(k=5,10,15,20,25)$ for the six real-world hypergraphs.}
  \label{fig4}
\end{figure}
%
As shown in Figure \ref{fig4}, the overlap between HEC and the classical centralities is generally low, with the most pronounced differences observed when focusing on the smallest sets of critical vertices ($k = 5, 10$), where Jaccard indices mostly far below the high overlaps observed among classical centralities themselves (e.g., DC-HDC, CC-VC).
As $k$ increases, the overlap between HEC and measures such as DC and HDC gradually rises to $0.5-0.7$, reflecting a convergence of vertex importance rankings across different centralities for broader ranges, while the very few most central vertices tend to be identified by multiple measures.

Overall,  HEC consistently exhibits relatively higher overlap with HDC and VC (especially for $k \geq 15$), whereas its overlap with CC remains persistently lower.
Notably, these overlap patterns are highly consistent across six networks with diverse topological characteristics, confirming the robust discriminative ability of HEC across different hypergraph structures.
In summary, as a novel vertex centrality, HEC identifies sets of critical vertices that are substantially different from those highlighted by classical centralities, demonstrating exceptional uniqueness in pinpointing the most influential vertices, and thereby offers a fresh perspective for vertex importance assessment in hypergraphs.

\section{Conclusion}
In this paper, we proposed a novel eigenvector centrality for hypergraphs, denoted by HEC.
The main new idea is to build an adjacency tensor for hypergraphs, which enables us to formulate a tensor eigenvalue problem whose positive solution yields the centrality scores of vertices.
For uniform hypergraphs, HEC reduces to the eigenvector centrality introduced by Benson. 
Moreover, in the limiting case where every hyperedge contains exactly two vertices, HEC further reduces to Bonacich's eigenvector centrality.
The proposed HEC captures higher-order interactions by aggregating the geometric mean of the scores of all vertices within each hyperedge.
Experimental results show that HEC shares certain correlations with classical centrality measures (e.g., degree, hyperdegree, clique-expansion, and vector centralities) in some networks, indicating a degree of consistency in identifying generally important vertices. However, HEC identifies critical vertices are often different from those highlighted by traditional measures, especially among top-ranked vertices, and exhibits superior robustness.
Importantly, vertices in smaller hyperedges tend to receive higher centrality because their scores depend more directly on a few neighbors, whereas larger hyperedges dilute the influence.
This enables precise identification of core vertices within small, dense groups while avoiding the spurious high influence induced by large hyperedges.

In summary, HEC provides a fresh perspective for assessing vertex importance in hypergraphs by jointly considering hyperedge heterogeneity, local influence, and global iterative. It complements existing centrality measures and offers a unique perspective.

Future work includes developing more efficient algorithms for computing HEC in very large hypergraphs to overcome the dimensionality issue of the current tensor-based approach, extending the framework to directed hypergraphs, weighted hypergraphs, or hypergraphs with temporal dynamics, investigating the spectral properties of hypergraphs within the proposed tensor framework and their applications in network analysis, and exploring its applications in specific domains such as recommendation systems, biological network analysis, and so on.
\section*{CRediT authorship contribution statement}
\textbf{Changjiang Bu:} Conceptualization,  Writing-editing. \textbf{Haotian Zeng:}  Conceptualization, Methodology, Writing-original draft. \textbf{Qingying Zhang:} Conceptualization, Methodology, Writing-original draft, Visualization, Validation.
The order of authors follows the alphabetical order of their surnames.
All authors have read and approved the final version of the manuscript.
\section*{Declaration of competing interest}
The authors declare that they have no known competing financial interests or personal relationships that could have appeared to influence the work reported in this paper.
\section*{Acknowledgments}
The research of the third author is partially supported by the National Natural Science Foundation of China (No.12071097) and the Natural Science Foundation for The Excellent Youth Scholars of the Heilongjiang Province (No. YQ2022A002).

\section*{References}
\bibliographystyle{plain}
\bibliography{ml0ht2}
\end{spacing}
\end{document}